\documentclass[11pt]{amsart}
\usepackage{graphicx} 

\usepackage{comment}
\usepackage{subcaption}
\usepackage{amsmath}
\usepackage{amssymb}
\usepackage{amsthm}
\usepackage{graphicx}
\usepackage{mathtools}
\usepackage{url}
\usepackage{enumerate}
\usepackage{color}
\usepackage{hyperref}
\usepackage{esint}

\usepackage[margin=1.0in]{geometry}

\newtheorem{theorem}{Theorem}

\newtheorem{proposition}[theorem]{Proposition}
\newtheorem{lemma}[theorem]{Lemma}

\newtheorem{corollary}[theorem]{Corollary}

\newtheorem{fact}[theorem]{Fact}

\newtheorem{prop}[theorem]{Proposition}

\theoremstyle{definition}
\newtheorem{definition}[theorem]{Definition}
\newtheorem{assumption}[theorem]{Assumption}
\newtheorem{question}[theorem]{Question}

\theoremstyle{remark}
\newtheorem{remark}[theorem]{Remark}

               \def\bfP{{\bf{P}}}          

\renewcommand{\H}{\mathbb{H}}
\newcommand{\G}{\Gamma}
\newcommand{\s}{\sigma}

\newcommand{\EE}{\mathbb{E}}
\newcommand{\PP}{\mathrm{P}}
\newcommand{\R}{\mathbb{R}}
\newcommand{\Z}{\mathbb{Z}}
\newcommand{\bN}{\mathbb{N}}
\newcommand{\M}{\mathbb{M}}
\newcommand{\bX}{\mathbb{X}}

\newcommand{\Mf}{\cP_{\mathsf{f}}}

\newcommand{\cO}{\mathcal{O}}
\newcommand{\cU}{\mathcal{U}}
\newcommand{\cP}{\mathtt{PM}}
\newcommand{\Closed}{\mathtt{Closed}}

\newcommand{\ann}{\operatorname{ann}}

\newcommand{\Hom}{\operatorname{Hom}}

\newcommand{\Ind}{\operatorname{Ind}}
\newcommand{\isom}{\operatorname{isom}}

\newcommand{\Isom}{\operatorname{Isom}}
\newcommand{\Sym}{\operatorname{Sym}}
\newcommand{\sym}{\operatorname{Sym}}
\newcommand{\density}{\operatorname{density}}
\newcommand{\opt}{\operatorname{opt}}
\newcommand{\per}{\operatorname{per}}
\newcommand{\Pois}{\operatorname{Pois}}
\newcommand{\Prob}{\operatorname{Prob}}
\newcommand{\Stab}{\operatorname{Stab}}
\newcommand{\Dopt}{D_{\opt}}
\newcommand{\Dper}{D_{\per}}
\newcommand{\DPois}{D_{\Pois}}
\newcommand{\Vol}{\operatorname{Vol}}
\newcommand{\vol}{\operatorname{Vol}}
\newcommand{\lam}{\lambda}
\newcommand{\Lam}{\Lambda}
\newcommand{\mes}{\mathrm{Vol}}
\newcommand{\dist}{\mathrm{dist}}
\newcommand{\eps}{\varepsilon}
\renewcommand{\epsilon}{\eps}

\newcommand{\one}{\mathbf{1}}

\newcommand{\Gibbs}{\operatorname{Gibbs}}

\renewcommand{\P}{\mathbb{P}}

\def\cc{{\curvearrowright}}

\def\bfPi{{\mbox{\boldmath $\Pi$}}}

\title{A phase transition for the hard sphere model on the hyperbolic plane}

\begin{document}

\author{Lewis Bowen}
\address{University of Texas at Austin, Department of Mathematics}
\email{lpbowen@math.utexas.edu}

\author{Marcus Michelen}
\address{Northwestern University, Department of Mathematics}
\email{michelen@northwestern.edu}

\author{Will Perkins}
\address{Georgia Institute of Technology, School of Computer Science}
\email{math@willperkins.org}

\begin{abstract}
    The hard sphere model is a classical model from statistical physics in which particles are represented by equal-sized spheres.  Longstanding predictions from the physics literature indicate that in $\R^2$ and $\R^3$ the system undergoes a phase transition, but it remains a major open problem to confirm this.  We prove the existence of a phase transition for this model in the hyperbolic plane.  
\end{abstract}
\maketitle

\section{Introduction}
\label{secIntro}

The hard sphere model is a classical model of a gas from statistical physics in which particles are represented by equal-sized spheres; the only interaction between particles is a hard-core exclusion: two spheres are not allowed to overlap.  

Classically, the model is studied in the physically relevant settings of $2$- and $3$-dimensional Euclidean space.   In $\R^3$ the model is expected to undergo a crystallization phase transition~\cite{alder1957phase,wood1957preliminary}, representing the freezing of a gas to a solid; in $\R^2$ a different type of phase transition is expected, with a breaking of rotational symmetry rather than translational~\cite{Krauth,richthammer2007translation}.  The hard sphere model has close connections to the classical sphere packing problem: determining the densest possible packing of equal-sized spheres in Euclidean space.  Maximum packings are the ground states for the hard sphere model, and  crystallization occurs when a random sphere packing of sub-optimal density aligns with a ground state. The study of the hard sphere model and the sphere packing problem in higher dimensions or in other underlying spaces is of interest both in physics and mathematics for several reasons including understanding the phase transition phenomenon; drawing analogies to  models of glasses; and studying idealized models of error-correcting codes~\cite{charbonneau2021three,cohn2016sphere,parisi2010mean}. 
However, despite over a century of work on the model (dating back at least as far as the work of Boltzmann), the existence of a phase transition in the hard sphere model in $\R^d$ has not been proved, for any dimension $d$. See~\cite{lowen1994melting,lowen2000fun,mulero2008theory} for background on the model.

Our main result is that the hard sphere model on the hyperbolic plane $\H^2$  exhibits a phase transition.  

We now define the model and the notion of a phase transition.  We begin by defining the hard sphere model in finite volume, i.e. on a bounded region in $\R^d$, $\H^d$, or  some other metric measure space of finite measure.

Let $(\bX, \dist, \mes)$ be a metric measure space and $\Lambda \subset \bX$ with $\mes(\Lambda)<\infty$. For $\lam > 0$ and $r > 0$, \textbf{the hard sphere model} on $\Lambda$ at activity $\lam$ and radius $r$ is the probability distribution on  finite point sets $X \subset \bX$ given by taking a Poisson point process $X$ of intensity $\lam$ on $\Lam$ and conditioning on the event that $\dist(x,y)\ge 2r$ for all distinct $x, y \in X$.  Let $\mu_{\Lambda, r,\lam}$ denote the law of the hard sphere model.  The support of $\mu_{\Lambda, r,\lam}$ is the set of all  point sets $X\subset \Lambda$ so that all pairwise distances are at least $2r$. We call this set $\cP_r(\Lambda)$ (the space of $2r$-separated point measures on $\Lambda$), noting that any such $X$ is the set of centers of a packing of interior-disjoint spheres of radius $r$ in $\Lambda$ (where only the centers of the spheres are  required to lie  in $\Lambda$).

The radius parameter $r$ governs the interaction distance (and size of the spheres) while the activity parameter (or ``chemical potential'' or ``fugacity'') $\lam$ governs the typical size of $X$.

If the space $\bX$ has infinite volume then we need a different approach to define the hard sphere model on $\bX$; this is accomplished via the Dobrushin-Lanford-Ruelle (DLR) equations~\cite{dobruschin1968description,lanford1969observables}.   We say a probability measure $\mu$ on $\cP_r(\bX)$ is {\bf Gibbs for the hard sphere model at radius $r$ and activity $\lambda$} if for every finite-volume measurable subset $\Lambda \subset \bX$, if $\bfP$ is a random sample of $\mu$ then the law of $\bfP$ restricted to $\Lambda$ conditioned on $\bfP \cap \Lambda^c$ is the finite volume  hard sphere model at activity $\lam$ on the set $\Lambda \setminus \{ x \in \Lambda: \dist (x, \bfP \cap \Lambda^c) <2r\}$.  In other words, the conditional distribution of $\mu$ restricted to $\Lambda$ given the realization of the points in $\Lambda^c$ is the finite-volume hard sphere model on $\Lambda$ with {\bf boundary conditions} imposed by the points  $\bfP \cap \Lambda^c$.   For brevity, we will say such a measure is $(r,\lambda)$-Gibbs.

A compactness argument in the space of probability measures on $\cP_r(\bX)$ shows there always exists at least one $(r,\lambda)$-Gibbs measure on $\bX$.  The main question on the topic is whether or not there is a unique Gibbs measure for a given $\bX, r, \lam$.  A {\bf phase transition} is the transition, as $\lam$ changes with $\bX$ and $r$ fixed, between uniqueness and non-uniqueness of Gibbs measure.  See e.g.~\cite{ruelle1999statistical,georgii1997stochastic,friedli2017statistical} for more   discussion.  Under fairly general conditions, a perturbative approach shows that for small $\lambda$ there is uniqueness of Gibbs measure (see e.g.\ \cite{jansen2019cluster} for a survey at a wide level of generality).  We introduce a threshold in $\lambda$ to capture a change from uniqueness to non-uniqueness:

\begin{definition}
Let $\lambda_{u}(\bX,r) \in [0,\infty]$ be the infimum over all activities $\lambda$ such that there are at least two distinct $(r,\lam)$-Gibbs measure for the hard sphere model on $\bX$. 
\end{definition}

We say a {\bf phase transition occurs in the hard sphere model on $\bX$ with spheres of radius $r$ if} $\lambda_u(\bX,r) \in (0,\infty)$.  The fact that there is uniqueness of Gibbs measure for small $\lambda$ assures that $\lambda_u(\bX,r) > 0$ and so the main work is to show $\lambda_u(\bX,r) < \infty$.
We note that repulsive models such as the hard sphere model need not obey monotonicity properties as $\lambda$ increases.  In particular, it may be the case that as $\lambda$ increases, there are multiple changes from uniqueness to non-uniqueness and back; examples where there is more than one phase transition have been constructed for analogous discrete models \cite{brightwell1999nonmonotonic}.

In Euclidean spaces, the question of existence of a phase transition  does not depend on the radius $r$ due to scale invariance, but in other spaces like $\H^d$ the answer might depend on $r$.    Currently only lower bounds on $\lam_u(\R^d,r)$ are known,  and it is a major open question whether $\lambda_u(\R^d, r)$ is finite for any dimension $d$.   There is no phase transition in dimension one: $\lam_u(\R,r)=\infty$. {The best bounds on absence of phase transitions in Euclidean space are given by~\cite{gobel2026uniquenessanalyticitymixinggibbs}; in greater generality (including $\H^d$), the best bounds are given by~\cite{michelen2021potential}.  A general bound that uses very little geometry is that $\lam_u(\bX,r) \ge e/ v(\bX,2r)$, where $v(\bX,2r)$ is  the volume of the ball of radius $2r$ in $\bX$~\cite{michelen2020analyticity,michelen2021potential}.}

The main result of this paper is that  there is an open unbounded set of radii $r$ such that $\lambda_u(\H^2, r)$ is finite; that is, a phase transition occurs in the hard sphere model on the hyperbolic plane.

\begin{theorem}\label{T:main}

There exists an unbounded open set $R_2 \subset (0,\infty)$ such that for every $r \in R_2$, $\lambda_{u}(\H^2,r) < \infty$. In fact, $R_2$ contains an interval of the form $(\rho,\infty)$ for some $\rho<\infty$.
\end{theorem}

At a  high level, the proof of Theorem~\ref{T:main} proceeds by constructing two  $(r,\lam)$-Gibbs measures on $\H^d$ with different {\bf densities}.  Finding the right notion of density to use is essential to the approach and  is not entirely straightforward as there are complications to defining densities of sphere packings in hyperbolic space due to its non-amenability. 

To define the relevant notion of density, we follow the approach of the first author and Radin \cite{bowen2002densest}: rather than looking just at sphere packings, we look at isometry-invariant \emph{probability measures} on sphere packings.  The density is then the probability that a random packing sampled from this distribution covers a given point (see Section \ref{sec:optimal-packing-density}).  Our approach to showing a phase transition demonstrates that for certain radii, at large activity there are two Gibbs measures of different densities, which are thus distinct. 

The two Gibbs measures we construct can be thought of as a distribution of random-like packings and a distribution of  more structured packings respectively.  To generate the more structured packings, we note that in the hyperbolic plane at certain radii $\{r_n\}$ for $n \in \{7,8,9,\ldots\} \cup \{\infty\}$ there is a unique packing of maximum density and these packings are in fact lattice packings (see Sections \ref{sec:why-dim-2} and \ref{sec:packings} for their definitions, and Figure \ref{fig:optimal-packings} for figures).  These special radii $r_n$ are called \textbf{tight} radii.  For $r$ near a tight radius, we take quotients of $\H^2$ to yield hyperbolic surfaces whose optimum packing density is close to the optimum on $\H^2$; we then may take the hard sphere model on these finite volume surfaces, lift them to $\H^2$ and take a subsequential limit to obtain a Gibbs measure of near-optimal density.  This is accomplished in Section \ref{secDenseGibbs}.

To generate  random-like packings, we start with the empty configuration on $\H^2$ and run \textbf{spatial birth-death Glauber dynamics}, a continuous-time Markov chain that in finite volume has stationary distribution given by the hard sphere model.  Roughly, these dynamics attempt to add spheres randomly according to a Poisson process and remove them according to independent exponential clocks of rate $1$.   Running this process up to time $T$ yields an isometry-invariant probability measure $\mu_T$ on packings.  By taking a subsequence of $(\mu_T)_{T > 0}$ we obtain a limiting measure $\mu_\ast$ that satisfies the Georgii-Nguyen-Zessin (GNZ) equations and thus is a Gibbs measure (see Section \ref{sec:Gibbs-definitions} for preliminaries on Gibbs measures).  This is accomplished in Section \ref{sec:lower-density}. 

We now are faced with a seemingly innocuous task: how may we distinguish the ``random-like'' measure created via Glauber dynamics and the ``structured'' measure created from the lattice packing?  To begin with, we note that our ``random-like'' measure is a \textbf{weak Poisson factor}, namely a weak limit of probability measures given by an equivariant thinning of a (homogeneous) Poisson process (see Section \ref{sec:Poisson-factors} for a formal definition).  At the tight radii $\{r_n\}$, we recall that there is a unique probability measure $\mu_n$ on packings of optimum density.  Thus, to show that at these exact radii that  our ``random-like'' and ``structured'' packing are different, we  need to show that $\mu_n$ is not a weak Poisson factor. This will imply that all Poisson factors have density strictly less than the density of $\mu_n$, and so we just need to be sure our ``structured'' packing has density sufficiently close to the optimum in order to deduce it is distinct from our ``random-like'' packing.

In order to show that $\mu_n$ is not a weak Poisson factor, we use the theory of \textbf{annealed (sofic) entropy}, a notion of entropy for actions of certain non-amenable groups.  By using a non-amenable free group $F$ that lies within the isometries of $\H^2$, we will consider a probability measure preserving action of $F$ on $(M,\mu_n)$ where $M$ is a quotient of $\H^2$.  We will show that the periodic measure $\mu_n$ has annealed entropy equal to $-\infty$ while every weak Poisson factor has non-negative annealed entropy.  
This provides a quantitative basis for describing one Gibbs measure as ``random-like'' and the other as ``structured.''  

The easiest way to define annealed entropy is for an isometry-invariant probability measure $\mu$ on ${\tt A}^F$ where ${\tt A}$ is a finite set and $F$ is a non-amenable free group.  If $F$ has rank $r$, then the Cayley graph of $F$ with the usual generators is the infinite $2r$-regular tree.  Recall that random $2r$-regular graphs on $n$ vertices Benjamini-Schramm converge to the Cayley graph of $F$ as $n$ tends to infinity.  If we consider some $x: [n] \to {\tt A}$ and have a random $2r$-regular graph $G_n$, the empirical measure of $x$ may be thought of as the distribution on ${\tt A}^F$ where we place the root of $F$ at a random vertex in $G_n$.  The annealed entropy $h^{\mathrm{ann}}(\mu)$ is the exponential growth rate of the expected number of $x:[n] \to {\tt A}$ that approximate $\mu$ well, where the expectation is taken over random regular graphs.  We define these central quantities in Section \ref{sec:annealed-entropy}.

In our case, we will translate from a Poisson factor to a Bernoulli factor by considering a fundamental domain for $F$ in the group of isometries of $\H^2$.  We perform the reduction from our setting to the more classical setting of annealed entropy in Section \ref{sec:reduction-free}.  We then show that annealed entropy separates Bernoulli factors from periodic measures under an ergodicity assumption (Propositions \ref{P:nonnegative} and \ref{P:smoothactions}).

We briefly discuss the approach to densities we use, starting with an approach to the optimal sphere packing problem in hyperbolic space due to the first author and Radin~\cite{bowen2002densest}.

\subsection{Optimal packing density} \label{sec:optimal-packing-density}

Let $\bX$ be an infinite volume, locally compact homogeneous metric measure space, like $\R^d$ or $\H^d$.  Let $X \in \cP_r(\bX)$ be the set of centers of a packing of spheres of radius $r$.  The {\bf density} of the packing of spheres of radius $r$ around $X$ with respect to a basepoint $p\in \bX$ is 
$$\density_r(X,p) = \lim_{R \to \infty} \frac{\mes(B(p,R) \cap  \{y: \dist(y,X) \le r    \})}{\mes(B(p,R))}$$
where $B(p,R)$ is the closed  ball of radius $R$ centered at $p$. Of course, this limit need not exist.  If $\bX=\R^n$ and the limit exists then it does not depend on the choice of basepoint $p$. However, if $\bX=\H^n$  the limit may depend on $p$. 

An old and venerated problem asks for the maximum density of a  sphere packing of Euclidean space~\cite{conway2013sphere,cohn2017conceptual}. There has been exciting recent progress in a few specific dimensions \cite{viazovska2017sphere,cohn2016sphere} (optimal packing densities in dimensions $1,2,3,8$ and $24$ are now known)  and in the case of large dimensions~\cite{campos2023new,klartag2026lattice} (though the best upper and lower bounds still differ by an exponential factor in $d$). 

The analogous problem in hyperbolic space was stymied for a  time by the fact that density of a sphere packing may depend on the choice of basepoint (see, e.g.,~\cite{boroczky1978packing,toth1993packing,MR1604938}). To formulate a well-defined optimization problem, the first author and Radin focused attention on probability distributions on packings that are invariant under the isometries of $\H^d$~\cite{bowen2002densest}.

With $\bX \in \{ \R^d, \H^d\}$, define $\Prob(\cP_r(\bX))$ to be the space of all Borel probability measures on $\cP_r(\bX)$ and note that $\Prob(\cP_r(\bX))$ is compact\footnote{By the Riesz-Markov Theorem, we may consider $\Prob(\cP_r(\bX))$ as embedded in the Banach dual of $C(\cP_r(\bX))$, which is the space of continuous functions on $\cP_r(\bX)$. We endow the dual  $C(\cP_r(\bX))^*$ with the weak* topology. Thus a sequence $(\nu_j)_{j \geq 1}$ of measures converges to $\nu$ if for every continuous function $f:\cP_r(\bX) \to \R$ we have $\nu_j(f) \to \nu(f)$ where $\nu_j(f)=\int f~d\nu_j$.  It is an exercise to check that $\Prob(\cP_r(\bX))$ is closed in the unit ball.  So the Banach-Alaoglu Theorem implies $\Prob(\cP_r(\bX))$ is compact.}. 
Let $\M_r(\bX)\subset \Prob(\cP_r(\bX))$ be the space of Borel probability measures on $\cP_r(\bX)$ that are invariant under $\Isom(\bX)$, the group of isometries of $\bX$.  In particular, this is the set of probability measures $\nu$ supported on $\cP_r(\bX)$ so that for all Borel sets $E \subset \cP_r(\bX)$ and all $g \in \Isom(\bX)$ we have $\nu(gE) =\nu(E)$. This space is closed in $\Prob(\cP_r(\bX))$ and therefore is weak* compact.

Fix a point $\cO \in \bX$ which we call the origin.  The {\bf density} of an invariant measure $\mu \in \M_r(\bX)$ is the probability that the spheres of radius $r$ centered around  a random sample $X$ from $\mu$ cover the origin. That is,
$$\density_r(\mu) = \mu( \dist (\cO, X) \le r) \, .$$

The first author and Radin then define the {\bf optimal packing density} as the greatest density of $\mu\in \M_r(\bX)$:
$$\Dopt(\bX,r) := \sup\{ \density_r(\nu) : \nu \in \M_r(\bX)\}\,. $$

Assisted by an ergodic theorem due to Nevo and Stein~\cite{nevo1997analogs}, the first author and Radin proved: given $\mu \in \M_r(\H^d)$, for $\mu$-a.e.\ $X \in \cP_r(\H^d)$, $\density_r(X,p)$ exists and does not depend on the choice of basepoint $p$~\cite{bowen2002densest}. If $\mu$ is ergodic (that is, not a non-trivial convex combination of other measures in $\M_r(\bX)$), then $\density_r(X,p) = \density_r(\mu)$ for $\mu$-a.e.\ $X$. Moreover, the functional $\density_r: \M_r(\bX) \to \R$ is continuous (this is because the boundary of a disk has measure zero with respect to the $\Isom(\H^d)$-invariant measure on $\H^d$). Since $\M_r(\bX)$ is compact, the optimal density is achieved.  When $\bX = \R^d$, $\Dopt(\bX,r)$ coincides with the usual definition of the optimal sphere packing density (and does not depend on $r$), and so this is a natural generalization.

\subsection{Periodic packings}

We now specialize to $\bX = \H^d$ (and eventually to $\bX = \H^2$) and discuss a special family of packings, {\bf periodic packings}.

For a packing $P \in \cP_r(\H^d)$ let $\Sym(P) =\{g\in \Isom(\H^d):~gP=P\}\leq \Isom(\H^d)$ be its group of symmetries. If $P$ contains at least two distinct points then this group is discrete.
If $\Sym(P)$ is co-finite then we say that $P$ is {\bf periodic}.  For such a periodic packing, its orbit $O(P) := \Isom(\H^d) P$ is equivariantly homeomorphic with $\Isom(\H^d) / \Sym(P)$ and by co-finiteness of $\Sym(P)$ the Haar measure on $\Isom(\H^d) / \Sym(P)$ pulls back to a finite measure $\nu$ which we may normalize so that $\nu(O(P)) = 1$.  Extend the measure $\nu$ to $\cP_r(\H^d)$ by defining $\nu(\cP_r(\H^d) \setminus O(P)) = 0$.  A measure $\nu$ obtained in this fashion from a periodic packing is called a {\bf periodic measure}; we then extend this definition to say that $\nu \in \M_r(\H^d)$ is periodic if it is in the convex hull of all such measures.    The optimal periodic density is defined via $$\Dper(\H^d,r) := \sup\{\density_r(\nu) : \nu \in \M_r(\H^d) \text{ is periodic} \}\,. $$
Clearly $\Dper(\H^d,r) \leq \Dopt(\H^d,r)$.

\subsection{Why dimension \texorpdfstring{$2$}{2}?} \label{sec:why-dim-2}
To prove Theorem~\ref{T:main} we will use some special facts about optimal packing measures and periodic packings in $\H^2$. In~\cite{bowen2003periodic} the first author proved the following.

\begin{theorem}[\cite{bowen2003periodic}] \label{thm:dim2-facts}
    For each $r >0$, the space of periodic measures in $\M_r(\H^2)$ is dense in $\M_r(\H^2)$. In particular, $\Dopt (\H^2,r)=\Dper(\H^2,r)$.  Additionally, $\Dopt (\H^2,r)$ is a continuous function of $r$. 
\end{theorem}

The analogous results for $\R^d$ follow easily from  amenability of $\R^d$. It is unknown whether either of the results in Theorem \ref{thm:dim2-facts} extend to $\H^d$ for $d>2$.

One can ask if $\Dopt(\H^2,r)$ is in fact achieved by a periodic measure (rather than just approached by a sequence of such measures).  The first author and Radin proved that except for a countable set of radii $r$, $\Dopt(\H^2,r)$ is \textit{not} achieved by a periodic measure~\cite{bowen2002densest}.  However, there are infinitely many special radii for which the optimal density is achieved by a lattice packing.  In particular, for $n \ge 7$, define $r_n$ via the equation
\begin{equation}
\label{eqrndef}
    \cosh(2r_n) = \cot(2\pi/n)\cot(\pi/n) \,.
\end{equation}

The radius $r_n$ is the smallest radius so that one can have precisely $n$ disjoint balls of radius $r_n$ tangent to a single ball of radius $r_n$.  Recalling that in $\R^2$ the optimal packing has each ball touching 6 other balls, one may view the optimal packings at these radii as hyperbolic analogues of the optimal packing in $\R^2$.  In particular, one has a unique optimizer here. 

\begin{figure}[htbp]
  \centering
  \begin{subfigure}{0.32\textwidth}
    \centering
    \includegraphics[width=\linewidth]{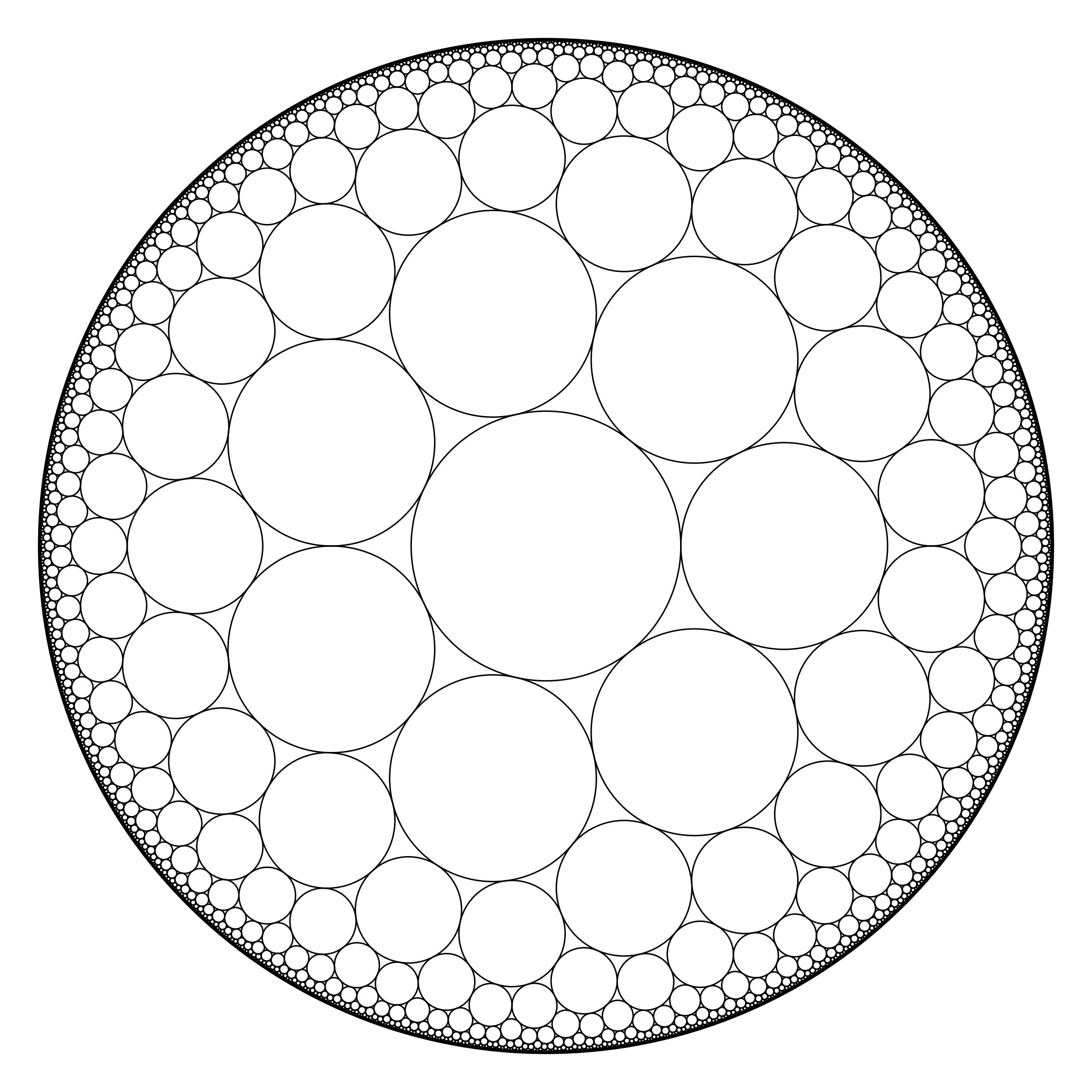}
  \end{subfigure}\hfill
  \begin{subfigure}{0.32\textwidth}
    \centering
    \includegraphics[width=\linewidth]{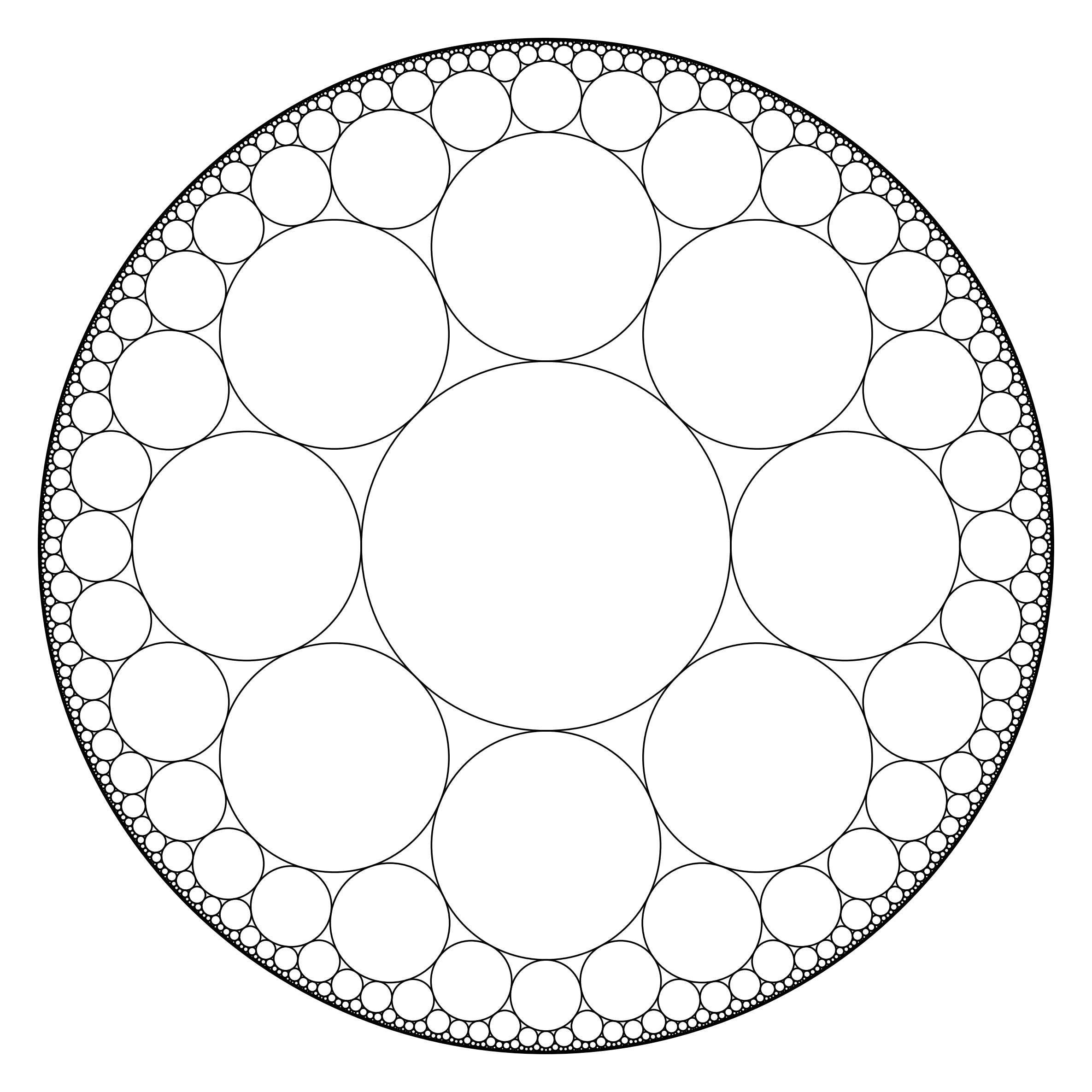}
  \end{subfigure}\hfill
  \begin{subfigure}{0.32\textwidth}
    \centering
    \includegraphics[width=\linewidth]{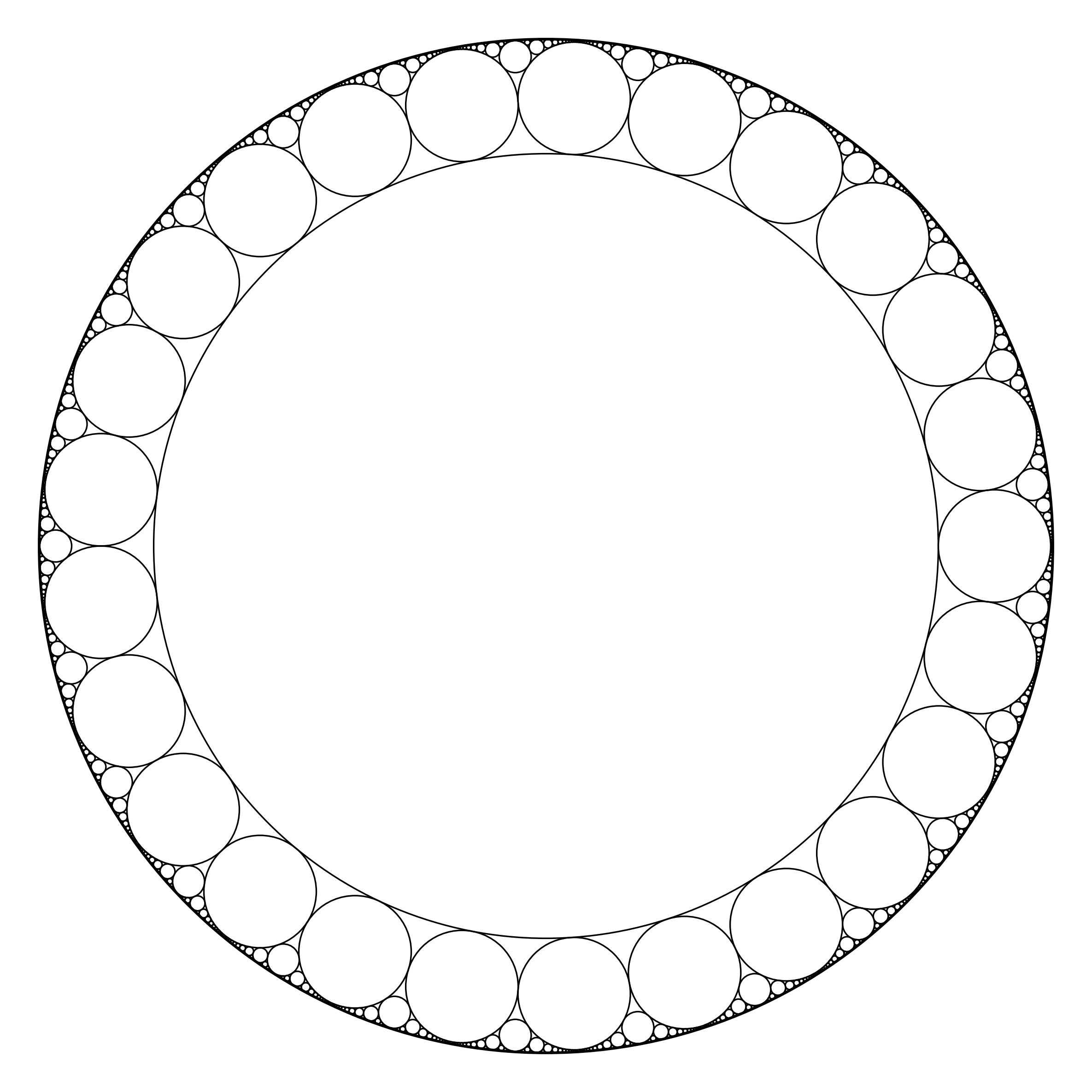}
  \end{subfigure}

  \caption{The optimal packings at radii $r_7, r_8$ and $r_{25}$}
  \label{fig:optimal-packings}
\end{figure}

\begin{theorem}[\cite{bowen2002densest}]\label{T:special}
    For each integer $n\ge7$ with $r_n$ defined via~\eqref{eqrndef} there is a unique  measure $\mu_n \in \M_{r_n}(\H^2)$ such that  $\density_{r_n}(\mu_n)=\Dopt(\H^2,r_n)$. Moreover, $\mu_n$ is a periodic measure.
\end{theorem}

We give further details below in Section~\ref{secLatticeMeasures}.   We also have an analogue of Theorem~\ref{T:special} for \textbf{horoball packings}, i.e.\ packings with $r = +\infty$; see Figure \ref{fig:horoball} for a figure of the optimal horoball packing.  We define horoballs more precisely in Section \ref{sec:horoballs} and describe the optimum in Theorem \ref{T:horoball1}. %
\begin{figure}[htbp]
  \centering
    \includegraphics[width=0.32\linewidth]{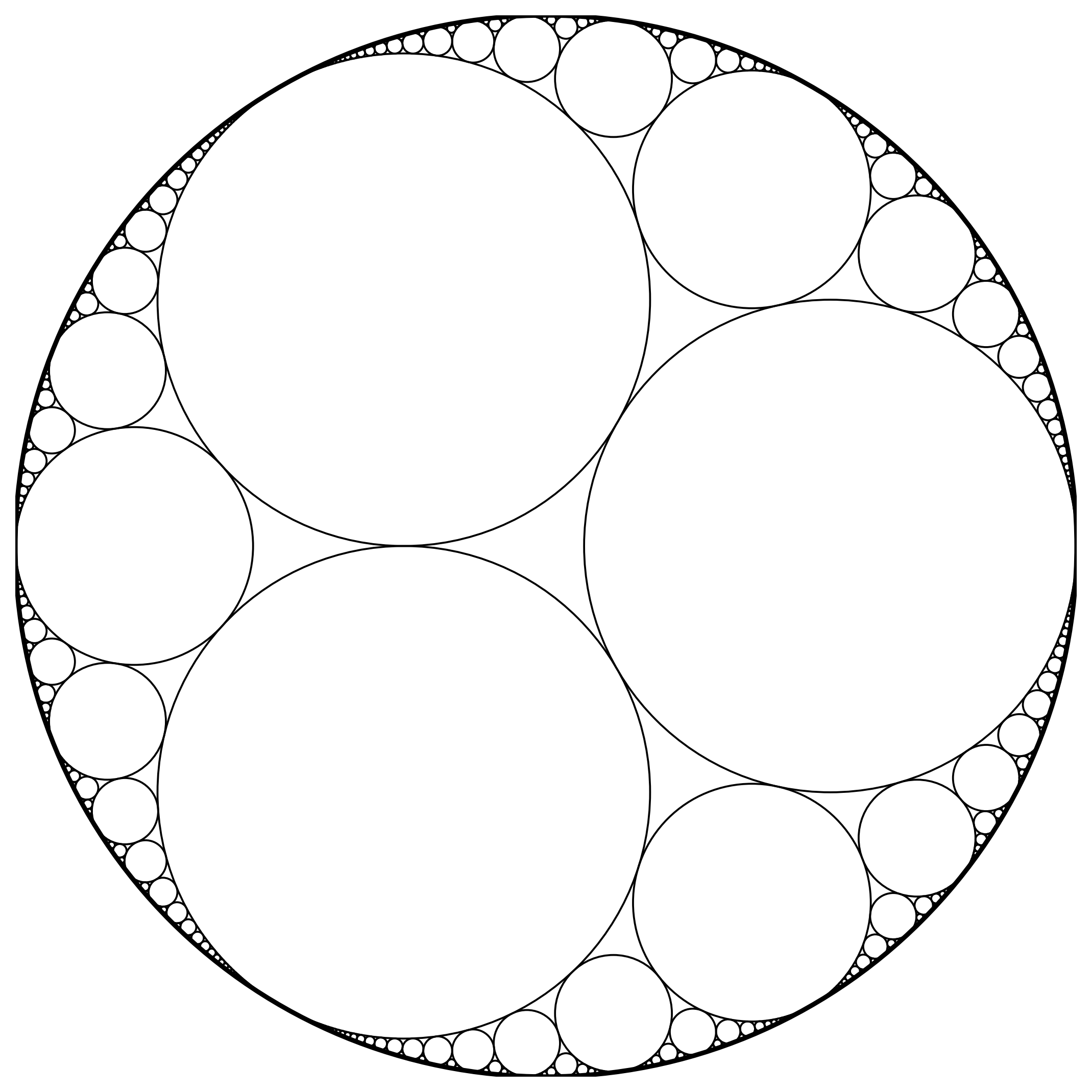}
  \caption{The optimal horoball packing}
  \label{fig:horoball}
\end{figure}

In conclusion, dimension $2$ is special because we know the existence of an unbounded set of radii for which the unique densest packing measure is  a periodic packing measure, the optimal density is continuous as a function of radius, and arbitrary invariant measures on the space of packings can be approximated by periodic measures. We do not know whether any of these results can be extended beyond dimension $2$.  

We will use the existence of optimal lattice packing measures to construct $(r,\lam)$-Gibbs measures with high density in Section~\ref{secDenseGibbs}.  We now introduce the background required for constructing our $(r,\lam)$-Gibbs measures with lower density.

\subsection{Poisson Factors}
\label{secPoissonFactorsIntro}

We will construct $(r,\lam)$-Gibbs measures whose density stays bounded away from the optimal density even as the intensity increases to infinity as weak limits of factors of Poisson processes.

For a parameter $T>0$ let $\Omega_T(\bX)$ denote the set of locally finite point sets in $\bX \times [0,T]$.  A {\bf Poisson process}  of intensity $1$ on $\bX \times [0,T]$ is a random element $\mathbf X$ from $\Omega_T(\bX)$ with the properties that:
\begin{enumerate}
    \item for each Borel set $B \subset \bX \times [0,T]$, the distribution of $|\mathbf{X} \cap B|$ is Poisson of mean $ \mes(B)$ where $\mes$ is the product measure of the volume measure on $\bX$ and the Lebesgue measure on $[0,T]$;
    \item for each $m$ and all disjoint Borel sets $B_1,\ldots,B_m$, the random variables $\{|\mathbf{X} \cap B_j|\}_{j = 1}^m$ are jointly independent.
\end{enumerate}
Let $\Pois_{\bX,T} \in \Prob(\Omega_T(\bX))$ denote the law of a Poisson process $\mathbf X$. This is uniquely characterized by the properties above.

Poisson processes  are prototypical models of non-interacting points in a space and are used to  define models in stochastic geometry and interacting point processes (including the hard sphere model). See, e.g.,~\cite{kingman1992poisson} for background.

A Poisson factor is one way to use a Poisson process as the underlying randomness in constructing another process.  In our setting, the following definition suffices.
\begin{definition}
   A measure on packings $\nu \in \M_r(\bX)$ is a {\bf Poisson factor} if there is some $T>0$ and some measurable function $\phi:\Omega_T(\bX)  \to \cP_r(\bX)$ so that
   \begin{enumerate}
       \item For all $g \in \Isom(\bX)$ we have $\phi(g X) = g \phi(X)$ where $g$ acts on $\Omega_T(\bX)$ and $\cP_r(\bX)$ in the natural way;
       \item $\nu = \phi_\ast \Pois_{\bX,T}$ ($\nu$ is the pushforward of the Poisson process under $\phi$).
   \end{enumerate}
    A measure $\nu \in \M_r(\bX)$ is a {\bf weak Poisson factor} if it is a weak limit of Poisson factors. 
\end{definition}

We define the optimal Poisson packing density $D_{\Pois}$ to be the highest density attainable by weak Poisson factors: 
\begin{equation}
    D_{\Pois}(\bX,r) := \sup\{\density_r(\nu) : \nu \in \M_r(\bX) \text{ is a Poisson factor}\}.
\end{equation}
While $D_{\Pois}(\bX,r)$ may not be achieved by a Poisson factor, by compactness of $\M_r(\bX)$ it is achieved by a weak Poisson factor.

In Euclidean space, using amenability we have $D_{\Pois}(\R^d,r) = \Dopt(\R^d,r)$, but we will see this need not be true in hyperbolic space.

\subsection{Proof strategy and further results}

For a given pair $(r, \lam)$ with $r$ close to some  tight radius $r_n$ defined at \eqref{eqrndef} and $\lam$ large, we will construct two different $(r,\lam)$-Gibbs measures on $\H^2$.
One of these measures will be constructed via a limiting process  beginning with the empty set of centers on $\H^2$ and then running continuous-time Glauber dynamics.  The approach will be to find a subsequential limit that satisfies the Georgii-Nguyen-Zessin (GNZ) equations, which essentially say that a measure is a Gibbs measure if and only if it is reversible with respect to Glauber dynamics.  We thus want to show that the ``GNZ defect''---which is simply the difference of the two sides of the GNZ equations---is zero for each test function.  The dynamics run up to finite time $T$ produces a Poisson factor $\nu_T$ (see Lemma \ref{lem:mu_t-Poisson-factor}). On each compact quotient, entropy dissipation bounds time average of the Fisher information (Lemma \ref{lem:FTC-entropy}), and Lemma \ref{lem:finite-volume-dissipation} shows that the GNZ defect in finite volume is bound in terms of the information per unit volume.   We then show that for a fixed test function on a fixed time scale, the GNZ defect and the finite-volume GNZ defect are close (Lemma \ref{lem:finite-propagation}).  We then may choose a sequence of times $t_j \in [j/2,j]$ along which the GNZ defect tends to zero and pass to a weakly convergence subsequence.  The limit satisfies the GNZ equations and is thus a Gibbs measure. This is performed in Section \ref{sec:lower-density}.  

The other Gibbs measure will be constructed as a limit of finite-volume Gibbs measures constructed from the optimal lattice sphere packings of $\H^2$ with radii $\{ r_n \}_{n \ge 7}$.  Since the optimal packings at radius $r_n$ are lattice packings, we may take subgroups $\{L_j\}$ of its lattice isometry group so that $L_j\backslash \H^2 $ Benjamini-Schramm converges to $\H^2$ (see Fact \ref{fact:residually-finite}); roughly, this means that for each fixed $R > 0$, for a random point in $L_j \backslash \H^2$, the ball of radius $R$ centered at that point looks like a radius $R$ ball in $\H^2$ (see Definition \ref{def:benjamini-schramm} for a precise definition of Benjamini-Schramm convergence).  The finite-volume surfaces $L_j \backslash \H^2$ will have packings of density $\Dopt(\H^2,r_n)$, and so taking the hard-sphere model on $L_j \backslash \H^2 $ at high activity will yield finite-volume Gibbs measures on packings of density approximately $\Dopt(\H^2,r_n)$ (see Lemma \ref{lem:packing-to-gibbs}).  Taking a subsequence of these finite-volume Gibbs measures as $j \to \infty$ will yield an $(r_n,\lambda)$-Gibbs measure (see Lemma \ref{lem:construct-infinite-volume-from-limits}).  Using continuity of $r \mapsto \Dopt(\H^2, r)$ along with the fact that $D_{\per}(\H^2,r) = \Dopt(\H^2,r)$ (see Theorem \ref{thm:dim2-facts}) we can in fact perform this approach for $r$ \emph{near} $r_n$ rather than just \emph{at} $r_n$.  This construction is performed in Section \ref{secDenseGibbs}.

To complete the proof, we show that lattice actions are not weak Poisson factors.  In particular, this will show that the lattice packings $\mu_n$  which are the unique optimal packings at the {tight} radii $\{r_n\}_{n \geq 7} \cup \{\infty\}$ are not weak Poisson factors.  To begin with, we will consider a periodic measure $\mu \in \M_\rho(\H^2)$ and take a fundamental domain of the underlying lattice in order to reduce to the more classical setting of Bernoulli factors in free groups (see Section \ref{sec:reduction-free} and Lemma \ref{lem:mc-to-Bernoulli-shift}).  By taking a free group $F$ in $\Isom(\H^2)$ we obtain an  action $F \cc (\cP_\rho(\H^2),\mu)$, which is ergodic by the Howe-Moore Theorem and has finite entropy by work of Rufus Bowen and Handel and Kitchens (see Proposition \ref{C:weak}).  Since this action has finite entropy, we then show the action $F \cc (\Isom(\H^2),\mu)$ is measurably conjugate to some action $(\mathtt{A}^F,\nu)$ for some finite set $\mathtt{A}$ and probability measure $\nu$, where $\nu$ is a weak Bernoulli factor if $\mu$ is (Proposition \ref{P:weakBernoullifactor1}).   We will review the theory of annealed (sofic) entropy and show that weak Bernoulli factors on $\mathtt{A}^F$ have non-negative annealed entropy (Proposition \ref{P:nonnegative}).  Conversely, we will see that finite entropy actions have annealed entropy equal to $-\infty$ (Proposition \ref{P:smoothactions}).  This will show that our original measure $\mu$ cannot simultaneously be supported on a lattice and a weak Bernoulli factor.  This is carried out in Section \ref{secLattice}.

As a consequence of uniqueness of the optimal measures at the tight radii $\{r_n\} \cup \infty$ along with properties of packings in $\H^2$, we  deduce the following main theorem:
\begin{theorem}
    \label{T:poissongap}
    There exists an unbounded open set $R \subset (0,\infty)$ such that for every $r \in R$, $D_{\Pois}(\H^2,r)<\Dopt(\H^2,r)$. Moreover, we have $\{r_n\}_{n \geq 7} \subset R$ and there is a real number $\rho$ such that $R \supset (\rho,\infty)$.
\end{theorem}

From here, non-uniqueness follows easily: for $r$ in the above set $R$, our nearly optimally dense Gibbs measure has density arbitrarily close to $\Dopt(\H^2,r)$  for sufficiently large activity $\lambda$; conversely, our Gibbs measure that is a weak Poisson factor has density bounded above by $D_{\Pois}(\H^2,r)$ \emph{for all} activities $\lambda$.  As such, taking $\lambda$ large enough to separate these densities yields multiple distinct Gibbs measures and thus our desired phase transition.

A second consequence of Theorem \ref{T:poissongap} is the existence of suboptimal sphere packings in $\H^2$ that cannot be made denser by removing and adding  only finitely many spheres while remaining a packing.  This is the notion of a \textit{completely saturated packing}, introduced in \cite{MR1604938} (see also \cite{bowen2003existence}).  Roughly, a packing by radius-$r$ balls is completely saturated if no finite number of centers can be removed from the packing and replaced with strictly more centers and remain a packing by radius-$r$ balls.

Here we will abuse terminology slightly by referring to elements of $\cP_r(\bX)$ as packings rather than as $2r$-separated point measures. 
\begin{definition}
A packing $P \in \cP_r(\bX)$ is {\bf completely saturated} if there does not exist  packings $P',P_0,P_1 \in \cP_r(\bX)$ satisfying
\begin{enumerate}
    \item $P$ is the disjoint union of $P'$ and $P_0$;
    \item $P_0$ and $P_1$ each contain only finitely many disks and $P_0$ contains fewer disks than $P_1$;
    \item $P'$ is disjoint from $P_1$ and $P' \cup P_1 \in \cP_r(\bX)$ is a packing.
\end{enumerate}
\end{definition}

We will prove that every weak Poisson factor $\mu \in \M_r(\H^d)$ achieving $\density(\mu) = D_{\Pois}(\H^d,r)$ is completely saturated.  In light of Theorem \ref{T:poissongap}, this implies that there are completely saturated $r$-packings of sub-optimal density in $\H^2$:

\begin{theorem}
    \label{T:saturated}
    Let $\mu \in \M_r(\H^d)$ be a measure which is a weak Poisson factor and $\density(\mu)=D_{\Pois}(\H^d,r)$. Then $\mu$-a.e. packing $P$ is completely saturated. In particular, if $d=2$ and $r \in R$ (where $R$ is as in Theorems~\ref{T:main},~\ref{T:poissongap}) then there exists an invariant measure supported on  completely saturated  packings with sub-optimal density.
\end{theorem}

The conclusions of both Theorem~\ref{T:poissongap} and Theorem~\ref{T:saturated} are false with $\H^2$ replaced by $\R^d$: in Euclidean space the density of a  factor of Poisson can approach the optimal density, and the only completely saturated packings are optimal packings.  In particular, our entire proof strategy for Theorem~\ref{T:main} makes essential use of the non-amenability of $\H^2$ and cannot apply to $\R^d$.

\subsection{Related work and discrete models}

While no phase transition has been proved for the hard sphere model on $\R^d$, the existence of phase transitions has been proved for a related discrete model, \emph{the hard-core model}, on certain infinite graphs.  We briefly describe some of the results here and use the hard-core model to illustrate some of the main ideas of our proof of Theorem~\ref{T:main}.

Suppose $G=(V,E)$ is a locally finite graph and $\lambda>0$. A subset $I \subset V$ is an {\bf independent set} if no edge in $E$ has both endpoints in $I$. Let $\Ind(G)$ denote the collection of independent sets of $G$.

If $G$ has only finitely many vertices then the {\bf hard-core model at activity $\lambda$} is 
a random subset $S \subseteq V$ with each vertex included independently with probability $\lam/(1+\lam)$ conditioned on the event $\{ S \in \Ind(G) \}$. If $G$ has infinitely many vertices then, as with the hard sphere model, we can define a $\lam$-hard-core model via the DLR equations and conditional measures. We define $\lam_u(G)$ as the infimum over activities $\lam$ so that there exist multiple distinct $\lam$-hard-core measures on $G$.

A phase transition for the hard-core model on $\Z^d$ with nearest neighbor edges, $d \ge 2$, (in the sense $\lam_u(\Z^d) < \infty$)  was proved by Dobrushin in 1968~\cite{dobrushin1968problem}. Dobrushin's upper bound on $\lam_u$ is exponentially large in $d$.  Galvin and Kahn later showed $\lambda_u(\Z^d) \leq d^{-1/4 + o(1)}$~\cite{galvin2004phase}, which was further improved to  $\lambda_u(\Z^d) \leq  d^{-1/3 + o(1)}$ by Peled and Samotij~\cite{peled2014odd}.  Very recently, Hadas and Peled \cite{hadas2026critical} showed $\lambda_u(\Z^d) = d^{-1 + o(1)}$ which matches the classical lower bounds of the form $\lam_u(\Z^d) = \Omega(d^{-1})$ up to sub-polynomial factors.

The hard-core model on the infinite, $\Delta$-regular tree $T_\Delta$ plays an important role in probability, algorithms, and computational complexity.  The Gibbs uniqueness threshold is $\lam_u(T_{\Delta}) = \frac{(\Delta-1)^{\Delta-1}}{(\Delta-2)^\Delta}$~\cite{spitzer1975markov,kelly1985stochastic}.  This threshold is a lower bound on $\lam_u(G)$ for all graphs  $G$ of maximum degree $\Delta$.  It also marks a computational threshold between the existence of polynomial-time algorithms and NP-hardness for the problem of sampling from the hard-core model on a graph of maximum degree $\Delta$~\cite{Weitz,sly2010computational}, as well as the threshold for rapid mixing of the Glauber dynamics for the same class of graphs~\cite{mossel2009hardness,anari2021spectral}.

The value of $\lam_u(T_\Delta)$ can be computed exactly by solving a fixed point equation, but we note that one can prove the existence of a phase transition on $T_\Delta$ in a softer fashion, especially for $\Delta$ large.  In order to demonstrate that on $T_\Delta$ there are at least two Gibbs measures for large $\lambda$---much like in the proof of Theorem \ref{T:main}---one can find two isometry-invariant Gibbs measures of different densities.  In particular, on the random $\Delta$-regular bipartite graph on $n$ vertices, the densest independent set has density $1/2 $.  For large $\lambda,$ the hard-core model on these graphs thus has density arbitrarily close to $1/2$; by mapping a random vertex to the root of $T_\Delta$ and taking weak (subsequential) limits,  one can lift this sequence  to a Gibbs measure on $T_\Delta$ of density arbitrarily close to $1/2$.  Conversely, if one considers the random $\Delta$-regular graph on $n$ vertices, the first-moment method shows that {with probability tending to $1$ as $n \to \infty$}, all independent sets have density uniformly bounded away from $1/2$.  We may then lift the hard-core model on these graphs to a Gibbs measure on $T_\Delta$ which necessarily has density bounded uniformly away from $1/2$.  

At a high level, this proof of a phase transition on $T_\Delta$ bears similarity to the proof of Theorem \ref{T:main}.  We indeed will lift a sequence of finite-volume Gibbs measures to obtain an infinite-volume Gibbs measure of nearly optimal density, as in the case of random regular bipartite graphs above.  However, rather than construct finite-volume approximations to $\H^2$ on which there are no packings of density near the optimum on $\H^2$, we show that the unique optimal packing on $\H^2$ cannot be approximated by Poisson factors.  Indeed, one can perform the same argument on $T_\Delta$: there is a unique isometry invariant probability measure with density equal to the optimal density of $1/2$; one can show that this probability measure is not a weak limit of factors of i.i.d.\ either by computing the optimal density of random regular graphs or through the annealed entropy approach that we follow here.

We pause for a moment to point out that the existence of a phase transition for the hard-core model is quite a delicate question. In particular, while Dobrushin proved that on $\Z^d$ with nearest neighbor edges there is a phase transition for $d \geq 2$, a foundational work by Heilmann-Leib proved that if one instead considers any line graph---i.e.\ where one constructs a new graph $H$ from the old graph $G$ by having the vertex set of $H$ equal to the set of edges of $G$ and having two vertices in $H$ be connected by an edge if the corresponding edges in $G$ shared a vertex---then uniqueness holds \emph{for all} $\lambda > 0$, i.e.\ $\lambda_u = +\infty$.  For vertex transitive graphs like $\Z^d$, its line graph is quasi-isometric to the underlying graph.  This shows, in particular, that the existence of a phase transition for the hard-core model depends on more than just large-scale geometric information.

\subsection{Outline of the paper}

\begin{itemize}
   
    \item \S \ref{secPrelim} collects definitions and auxiliary results needed in later sections.

    \item \S \ref{sec:packings} explains what is known about optimally dense sphere and horoball packings in hyperbolic space and the functions $D_{\opt}(\H^n,\cdot), D_{\per}(\H^n,\cdot)$.
    
    \item \S \ref{secDenseGibbs} shows that for large $\lam$, there are $(r,\lam)$-Gibbs measures with near optimal density. 

    \item \S \ref{sec:lower-density} shows that for every $(r,\lam)$, there is an $(r,\lam)$-Gibbs measure that is a weak Poisson factor.

    \item \S \ref{secLattice} shows the special lattice packings with radii $\{r_n\}$ cannot be realized as weak Poisson factors.

\item \S \ref{secPoissonGap} proves the main results: Theorems~\ref{T:poissongap},  \ref{T:main} and \ref{T:saturated} in this order, by combining earlier results.

\item \S \ref{S:open} is a collection of open problems. 
\end{itemize}

\section{Preliminaries}
\label{secPrelim}

Here we give precise definitions and state auxiliary results we will use later.  

\subsection{Spaces of point measures and basic properties}

\begin{definition}
Let $\bX$ be a locally compact second countable (lcsc) space. A measure $\Pi$ on $\bX$ is called a {\bf point measure} if it can be expressed as a sum of Dirac masses $\Pi = \sum_{x\in S} c(x)\delta_x$ for some locally finite discrete subset $S\subset \bX$ and non-negative integers $\{c(x)\}_{x \in S}$. 

The space of all point measures on $\bX$ is denoted $\cP(\bX)$. It is a subset of the space of Radon measures on $\bX$, which may be identified with a subset of the Banach dual $C_0(\bX)^\ast$ via the Riesz Representation Theorem. We will consider $C_0(\bX)^\ast$ with the weak$^*$ topology. Thus a sequence $(\Pi_j)_{j \geq 1}$ of point measures converges to $\Pi$ if for every compactly supported continuous function $f:X \to \R$ we have $\Pi_j(f) \to \Pi(f)$ where $\Pi_j(f) :=\int f~d\Pi_j$. Then $\cP(\bX)$ is a closed subset of $C_0(\bX)^\ast$ and therefore $\cP(\bX)$ is Polishable in the sense that there exists a complete separable metric inducing its topology.

If $G$ is a group with a jointly continuous action $G \times \bX \to \bX$ on $\bX$ then there is an induced (jointly continuous) action on $\cP(\bX)$ by 
$$h\left(\sum_{x\in S} c(x)\delta_x\right) = \sum_{x\in S} c(x)\delta_{hx}.$$ 
\end{definition}

In a packing of radius $r$ balls with disjoint interiors, all centers are mutually of distance at least $2r$ apart.  This motivates the following definition, which indicates when a point measure may be interpreted as an $r$-packing.

\begin{definition}
A point measure of the form $\Pi = \sum_{x\in S} c(x)\delta_x$ is {\bf $2r$-separated} if $S$ can be chosen so that $c(x)=1$ for all $x \in S$ and if $x,y \in S$ are distinct then the distance between them is at least $2r$. Let $\cP_r(\bX) \subset \cP(\bX)$ denote the subspace of $2r$-separated point measures. This space is compact in the weak* topology.
\end{definition}

We will make use of properties of probability measures on compact Hausdorff spaces and so we now introduce the relevant notation and topology.

\begin{definition}
In general, if $X$ is a compact Hausdorff space, then we let $\Prob(X)$ denote the space of Borel probability measures on $X$. By the Riesz-Markov Theorem, we may consider $\Prob(X)$ as embedded in the Banach dual of $C(X)$, which is the space of continuous functions on $X$ which we consider with the weak* topology. It is an exercise to check that $\Prob(X)$ is closed in the unit ball.  So the Banach-Alaoglu Theorem implies $\Prob(X)$ is compact.

Our focus will be on $\Prob(\cP_r(G/K))$ for some lcsc group $G$ and compact subgroup $K\le G$. This space admits a jointly continuous $G$-action. Indeed, the action of $G$ on $G/K$ by left-translations induces an action of $G$ on $\cP_r(G/K)$ (as explained above) which induces an action on $\Prob(\cP_r(G/K))$ by $(g\mu)(E)=\mu(g^{-1}E)$ for $g\in G, \mu \in \Prob(\cP_r(G/K))$ and $E\subset \cP_r(G/K)$. 
\end{definition}

Let $\M_r(\bX)\subset \Prob(\cP_r(\bX))$ be the space of Borel probability measures on $\cP_r(\bX)$ that are invariant under the isometries of $\bX$.  In particular, this is the set of probability measures $\nu$ supported on $\cP_r(\bX)$ so that for all Borel sets $E \subset \cP_r(\bX)$ and all $g \in \Isom(\bX)$ we have $\nu(gE) =\nu(E)$. This space is closed in $\Prob(\cP_r(\bX))$ and therefore is weak* compact.

If $\Pi$ is a point measure on $\H^d$ say then by abuse of notation we identify $\Pi$ with its support, which is a subset of $\H^d$. For example, the notation $p \in \Pi$ means $\Pi(p)>0$. 

\subsection{Lattices in \texorpdfstring{$\H^d$}{H^d} and Benjamini-Schramm convergence}

For a Lie group $G$, we recall that a \textbf{lattice} $L \leq G$ is a discrete subgroup of $G$ so that $G/L$ admits a finite $G$-invariant Borel measure. This measure is unique up to scaling because it is locally equal to left-Haar measure on $G$.

\begin{fact}\label{fact:lift}
    Let $L \leq \Isom(\H^d)$ be a lattice with quotient map $\pi:\H^d \to L \backslash \H^d$. For $\Pi \in \cP_r(L \backslash \H^d)$ let $\widetilde{\Pi} \in \cP_r(\H^d)$ be the lift defined as follows: if $\Pi=\sum_{x \in S} \delta_x$ for a discrete set $S \subset L \backslash \H^d$ then
    $$\widetilde{\Pi} = \sum_{x \in \pi^{-1}(S)} \delta_x.$$
    Then there is a unique isometry-invariant measure $\mu_\Pi \in  \M_r(\H^d)$ concentrated on the $\Isom(\H^d)$-orbit of $\widetilde{\Pi}$. Moreover, the map $\Pi \mapsto \mu_\Pi$ extends to an injective continuous affine map from $\Prob(\cP_r(L \backslash \H^d))$ to  $\M_r(\H^d)$.
\end{fact}
\begin{proof}
If $\Pi$ is the zero measure, then let $\mu_\Pi$ be the zero measure. Otherwise, the group $L' = \{g \in \Isom(\H^d):~g\widetilde{\Pi}=\widetilde{\Pi}\}$ is a lattice containing the lattice $L$. Typically $L'=L$, but if $\Pi$ is invariant under a nontrivial isometry of $L \backslash \H^d$, then $L'$ will be larger.

Let $O=\{g\tilde{\Pi}:~ g\in \Isom(\H^d)\} \subset \cP_r(\H^d)$ be the orbit of $\tilde{\Pi}$. Then the action of $\Isom(\H^d)$ on $O$ is topologically conjugate to the action of $\Isom(\H^d)$ on $\Isom(\H^d)/L'$. In particular, there is a unique $\Isom(\H^d)$-invariant Borel probability measure $\mu_\Pi$ supported on the orbit $O$. Thus $\mu_\Pi$ is ergodic.

It remains to show that $\Pi \mapsto \mu_\Pi$ extends to an injective continuous affine map.  For this, recall that a weak* compact convex subset $K$ of a dual Banach space is called a {\bf Choquet simplex} if for every point $x\in K$ there is a unique probability measure $\zeta_x$ supported on the set $K^{\operatorname{ex}}\subset K$ of extreme points such that $x$ is the barycenter of $\zeta_x$ in the sense that $x = \int y~d\zeta_x(y)$.  

The space $\cP_r(\H^d)$ of point measures embeds into $\Prob(\cP_r(\H^d))$ via $\Pi \mapsto \delta_\Pi$. Moreover, the image is exactly the space of extreme points of the Choquet simplex $\Prob(\cP_r(\H^d))$.

The map $\delta_\Pi \mapsto \mu_\Pi$ embeds the extreme points of $\Prob(\cP_r(L \backslash \H^d))$ into the extreme points of $\M_r(\H^d)$. This is because each $\mu_\Pi$ is ergodic and also, because the Ergodic Decomposition Theorem \cite{MR1784210} implies that the space $\M_r(\H^d)$ is a Choquet simplex.

Therefore, the map $\delta_\Pi \mapsto \mu_\Pi$ uniquely extends to an affine map from $\Prob(\cP_r(L \backslash \H^d ))$ to $\M_r(\H^d)$. It is injective because it is injective on the set of extreme points. It is continuous because the map $\delta_\Pi \mapsto \mu_\Pi$ is continuous.
\end{proof}

\begin{definition}\label{def:benjamini-schramm}
Recall that the injectivity radius of a manifold $M$ at a point $p\in M$ is half the minimum length of a homotopically nontrivial closed curve through $p$.      A sequence of finite-volume $d$-dimensional hyperbolic manifolds $M_1,M_2,\ldots$ {\bf Benjamini-Schramm converges} to $\H^d$ if for every radius $r>0$, the probability that a uniform random point in $M_n$ has injectivity radius at least $r$ tends to 1 as $n\to\infty$. (This is a special case of the more general concept of Benjamini-Schramm convergence of Riemannian manifolds for which we refer the reader to \cite{MR4520306}). 
\end{definition}

We will need the following basic fact about lattices in semi-simple Lie groups which is proved, for instance, on \cite[pg. 738]{abert2017on}.

\begin{fact}\label{fact:residually-finite}
    Let $L \leq \Isom(\H^d)$ be a lattice.  Then there are lattices $L_j \leq L$ of finite index (for  $j\in \mathbb{N}$) so that $L_j \backslash \H^d$ Benjamini-Schramm converges to $\H^d$ as $j\to\infty$.
\end{fact}

This follows from the fact that lattices in semi-simple Lie groups are finitely generated together with Malcev's theorem assuring that such finitely generated groups are residually finite.

\subsection{The hard sphere model: Gibbs measures in finite and infinite volume} \label{sec:Gibbs-definitions}

We recall that $\cP(\bX)$ is the space of all point measures on $\bX$.  Define $\Mf(\bX)$ to be the set of \emph{finite} point measures on $\bX$, i.e.\ measures $\gamma \in \cP(\bX)$ satisfying $\gamma(\bX) < \infty\,.$

For a choice of radius $r > 0$, the \textbf{hard sphere pair potential} at radius $r > 0$ is the function $\phi:\bX^2 \to \{0,+\infty\}$ given by 
\begin{equation*}
\phi(x,y) = \begin{cases}0 & \text{ if } \dist(x,y) \geq 2r \\
+\infty &\text{ if } \dist(x,y) < 2r\end{cases}\,.
\end{equation*}
For each $n$ define $H_n : \bX^n \to \{0,+\infty\}$ via \begin{equation*}
    H_n(x_1,\ldots,x_n) = \sum_{1 \leq i < j \leq n} \phi(x_i,x_j)
\end{equation*}
and setting $H_0 = 0, H_1 \equiv 0\,.$  

Given an $n$-tuple $\mathbf{x}=(x_1,\ldots, x_n) \in \bX^n$, let $\Phi(x)=\sum_{i=1}^n \delta_{x_i} \in \Mf(\bX)$. We view $\Phi$ as a continuous map from the disjoint union $\sqcup_{n \geq 0} \bX^n$ onto $\Mf(\bX)$. Define the {\bf energy} (or Hamiltonian) $H: \Mf(\bX) \to \{0,+\infty\}$ by setting $H(\Phi(\mathbf{x}))=H_n(\mathbf{x})$ for all $\mathbf{x} \in \bX^n$.

For a bounded measurable set $\Lambda \subset \bX$ and \textbf{activity} $\lambda \geq 0$ define the \textbf{grand canonical hard sphere partition function} $Z_\Lambda(\lambda)$ via \begin{equation}\label{eq:partition-function-def}
    Z_\Lambda(\lambda) = 1 + \sum_{n \geq 1} \frac{\lambda^n}{n!} \int_{\Lambda^n} \exp\left(-H\left( \sum_{j = 1}^n \delta_{x_j}\right) \right)\,d\mes^n(\mathbf{x})\,.
\end{equation}

The \textbf{finite volume hard sphere Gibbs measure} $\mu_{\Lambda}^\lambda \in \Prob(\cP(\Lambda))$ is then defined via 
\begin{equation}\label{eq:finite-vol-gibbs-def}
    \int_{\cP(\Lam)} f\,d\mu_{\Lam}^\lambda = \frac{1}{Z_\Lambda(\lambda)}\left(1 + \sum_{n \geq 1}\frac{\lambda^n}{n!}\int_{\Lambda^n} f\left(\sum_{j = 1}^n \delta_{x_j} \right) \exp\left(-H\left(\sum_{j = 1}^n \delta_{x_j} \right) \right)\,d\mes^n(\mathbf{x}) \right)
\end{equation}
for all measurable $f:\cP(\Lambda) \to [0,\infty).$  

We now define the hard sphere energy with boundary conditions.  For $\gamma \in \cP(\bX)$ and  measurable $S \subset \bX$ write $\gamma_S := \gamma|_S$ for the restriction of $\gamma$ to $S$. If  $H(\gamma) = 0$ and $\Lambda \subset \bX$ is bounded then define the conditional energy via $H_{\Lambda}(\eta \,|\, \gamma) := H(\eta + \gamma_{\Lambda^c})$. Define the partition function $Z_{\Lambda\,|\,\gamma}(\lambda)$ and the finite-volume Gibbs measure $\mu_{\Lambda\,|\,\gamma}^\lambda$ by replacing $H$ with $H_\Lambda(\bullet \,|\,\gamma)$ in lines \eqref{eq:partition-function-def} and \eqref{eq:finite-vol-gibbs-def} respectively.

An \textbf{infinite volume Gibbs measure} is a measure $\mu \in \Prob(\cP(\bX))$ supported on $r$-packings so that the \textbf{Dobrushin-Lanford-Ruelle} (DLR) equations: 
\begin{equation}\label{eq:DLR}
    \int_{\cP(\bX)} f d\mu = \int_{\cP(\bX)}  \int_{\cP(\Lambda)}f(\eta + \gamma_{\Lambda^c}) d\mu_{\Lambda\,|\,\gamma}^\lambda(\eta) \,d\mu(\gamma)
\end{equation}
hold for all bounded measurable $\Lambda \subset \bX$ and measurable $f:\cP(\bX) \to [0,\infty).$  We note that by a monotone class argument, to show \eqref{eq:DLR} holds for all measurable $f$ it is sufficient to prove it for bounded \textbf{local} functions, meaning functions $f$ for which there is some compact set $V$ so that $f(\eta)$ depends only on the restriction of $\eta$ to $V$. In fact, it is sufficient for \eqref{eq:DLR} to hold for all continuous local functions because an arbitrary bounded measurable local function can be approximated by continuous local functions in measure. 

We define $\Gibbs(\bX;r,\lambda) \subset \Prob(\cP_r(\bX))$ to be the set of Gibbs measures on $\bX$ for the hard sphere model with activity $\lambda$ and radius $r$. It is closed and therefore compact.  Classical results use the compactness of the space $\Prob(\cP_r(\bX))$  to show  $\Gibbs(\bX;r,\lambda)$ is non-empty (see e.g.\ \cite[Theorem 5.6]{jansen2018gibbsian} for a modern treatment applicable in our setting). 

We will also make use of an equivalent characterization of Gibbs measures, namely the Georgii-Nguyen-Zessin (GNZ) equations.  We say $\mu$ satisfies the \textbf{GNZ equations} if for all measurable $F:\bX \times \cP(\bX) \to [0,\infty]$ we have 
\begin{equation} \label{eq:GNZ-equations}
    \int_{\cP(\bX)} \sum_{x \in \eta}F(x,\eta - \delta_x) \,d\mu(\eta) = \lambda \int_{\bX} \int_{\cP(\bX)} F(x,\eta) e^{-H(x\,|\,\eta)} \,d\mu(\eta)\,d\Vol_\bX(x)\,.
\end{equation}
where we recall $H(x\,|\,\eta) = H(\eta + \delta_x)$ when $H(\eta) = 0$.
Georgii \cite{georgii1976canonical} and Nguyen-Zessin \cite{nguyen1979integral} showed that the GNZ equations are equivalent to the DLR equations (see \cite[Theorem 5.10]{jansen2018gibbsian} for a version in the necessary level of generality). 

We also will require use of the correlation functions of a point process.  
For a probability measure $\mu \in \Prob(\cP(\bX))$ and $k \in \mathbb{N}$, define its \textbf{factorial moment measure} $\mu^{(k)}$ as the  measure given by pushing forward $\mu$ by the map $\cP(\bX) \to \cP(\bX^k)$ given by $\sum_{j} \delta_{x_j} \mapsto \sum_{i_1,\ldots,i_k} \delta_{(x_{i_1},\ldots,x_{i_k})}$ where the sum is over distinct $k$-tuples of points. We say $\mu$ has \textbf{$k$-point correlation functions} $\rho_k:\bX^k \to \R$ defined $\Vol_\bX^{\otimes k}$-almost everywhere if $\rho_k$ is the density of the expected value of a random sample $\Pi \sim \mu^{(k)}$ with respect to $\Vol_\bX^{\otimes k}$.

\subsection{Poisson factors} \label{sec:Poisson-factors}

\begin{definition}
Let $\bX$ be an lcsc space with a Radon measure $\mu$. A {\bf Poisson point process on $\bX$ with intensity measure $\mu$} is a random variable $\bfPi$ taking values in $\cP(\bX)$ satisfying
\begin{enumerate}
    \item for any measurable $E \subset \bX$ with $\mu(E)<\infty$,  $\bfPi(E)$ is a Poisson random variable with mean $\mu(E)$;
    \item if $E_1,E_2,\ldots$ are pairwise disjoint measurable subsets of $\bX$ then the restrictions of $\bfPi$ to $E_1,E_2,\ldots$ are jointly independent random variables.
\end{enumerate}
All such processes have the same law, which we denote by $\Pois(\mu) \in \Prob(\cP(\bX))$.
\end{definition}

\begin{definition}
Fix an lcsc group $G$ and a jointly continuous action $G \cc X$ where $X$ is a compact metrizable space. A measure $\nu \in \Prob(X)$ is a \textbf{Poisson factor} if there exists a $G$-equivariant measurable map $\Phi:\cP(G) \to X$ such that $\nu = \Phi_*\Pois(\lambda)$ where $\lambda$ is a left-Haar measure. More precisely, we only require that $\Phi$ be defined on a subset of full measure with respect to $\Pois(\lambda)$. The equivariance of $\Phi$ means that $\Phi(g \Pi) = g \Phi(\Pi)$ for every $g\in G$ and $\Pois(\lambda)$-a.e. $\Pi \in \cP(G)$.

We say $\nu$ is a \textbf{weak Poisson factor} if there exists a sequence $\nu_1,\nu_2,\ldots \subset \Prob(X)$ such that each $\nu_i$ is a Poisson factor and $\lim_{i\to\infty} \nu_i = \nu$ in the weak$^*$ topology. 
\end{definition}

\begin{remark}
It might seem more natural to consider Poisson point processes on the hyperbolic plane $\H^2$, rather than Poisson point processes on $G=\Isom(\H^2)$. However, $\H^2$ can be identified with $G/K$ where $K$ is the stabilizer of a point in $\H^2$. Using this, we can obtain a $G$-factor map from Poisson point processes on $G$ to Poisson point processes on $G/K$.

To be precise, define $\Phi:\cP(G) \to \cP(G/K)$ by $$\Phi\left(\sum_{x\in S} \delta_x\right) = \sum_{x\in S} \delta_{xK}$$
where $S$ is a multi-set in $G$. This map is $G$-equivariant. Moreover, if $\bfPi$ is a Poisson point process on $G$ with intensity measure $\lambda_G$ (a left Haar measure on $G$) then $\Phi(\bfPi)$ is a Poisson point process on $G/K$ with intensity measure $\lambda_{G/K}$, which is the unique $G$-invariant measure on $G/K$ normalized so that if $O \subset G$ is right-$K$-invariant (meaning $OK=O$) then $\lambda_G(O)=\lambda_{G/K}(OK)$. 
\end{remark}

\section{Optimal Packings in Hyperbolic Space} \label{sec:packings}

\subsection{Optimal packings in \texorpdfstring{$\H^2$}{H^2}: constructing the periodic measures \texorpdfstring{$\mu_n$}{mu^n}}
\label{secLatticeMeasures}

Rogers~\cite{rogers1958packing} proved an upper bound on $\Dopt(\R^d)$ which was extended to spaces of constant curvature by B\"or\"oczky~\cite{boroczky1978packing}. To state it properly, let $S_r \subset \H^d$ be a regular simplex whose side-lengths all equal $2r$. In this context, regular means that for every permutation $\pi$ of the vertices of $S_r$ there is an isometry $\phi$ of $\H^d$ such that $\phi$ fixes $S_r$ and $\pi(v)=\phi(v)$ for every vertex $v$. Now let $S'_r \subset S_r$ be the subset consisting of all points $x\in S_r$ which are distance $\le r$ from the vertex set. B\"or\"oczky proved\footnote{Rogers' proof in Euclidean space \cite{rogers1958packing} provides a sharp upper bound for $\Dopt(\R^2) \leq \pi/\sqrt{12}$ and an exponential upper bound $\Dopt(\R^n) \leq (\sqrt{2} + o(1))^{-n}$ in high-dimensional Euclidean space.}: 
$$\Dopt(\H^d,r) \le \frac{\Vol(S'_r)}{\Vol(S_r)}.$$
This is also called the \textbf{simplex bound}.

In the hyperbolic plane $\H^2$, for any angle $\alpha \in [0,\pi/3)$, there are equilateral triangles with all interior angles $\alpha$. If $\alpha=0$ then such a triangle is ideal, meaning the vertices lie on the boundary $\partial \H^2$. For integers $n>6$, let $T_n \subset \H^2$ be an equilateral triangle whose interior angles all equal $2\pi/n$ and let $r_n$ be half the common side length of $T_n$. Via hyperbolic trigonometric formulas,  
\begin{equation} \label{eq:r_n-def}
\cosh(2r_n) = \cot(2\pi/n)\cot(\pi/n).
\end{equation}

There is a unique tiling of the plane by isometric copies of $T_n$ (up to shifting by isometries). This tiling is invariant under a symmetry group $L_n \le \Isom(\H^2)$ which acts transitively on the tiles. In particular, $L_n$ is a lattice subgroup.

The vertices of this tiling form a point set $P_n$ which is the center set for a radius $r_n$-packing. By construction, the density of this packing meets the simplex bound and is therefore optimal. In fact one can compute  \cite[p. 239]{MR0165423}
$$\Dopt(\H^2,r_n) = \frac{3 \csc(\pi/n)-6}{n-6}.$$
We may then define the isometry invariant periodic measure $\mu_n$ by applying Fact \ref{fact:lift}; by uniqueness of the packing (see e.g., the remark after Theorem 2.1 \cite{kellerhals1998ball}), $\mu_n$ is the unique element of $\M_r(\H^2)$ with $\density(\mu_n) = \Dopt(\H^2,r_n).$

\subsection{The space of closed subsets}

In the next sub-section, we will consider packings by horoballs which are balls of infinite radius. The formalism we have developed to this point does not apply because these are not obtained from point measures on $\H^d$. In order to handle these packings and finite radius packings on the same footing, we will use the space of closed subsets. 

Let $\Closed(\H^d)$ denote the collection of all closed subsets of $\H^d$.  A sequence $(F_n)_{n=1}^\infty \subset \Closed(\H^d)$ converges in the Fell topology to $F_\infty$ if for all $x\in F_\infty$ there exist $x_n\in F_n$ with $\lim x_n = x$ and for all sequences $(x_n)_{n=1}^\infty$ with $x_n \in F_n$ for all $n$ such that $\lim_n x_n = x_\infty$ exists, we must have $x_\infty \in F_\infty$. With this topology, $\Closed(\H^d)$ is compact, metrizable and $\Isom(\H^d)$ acts jointly continuously on $\Closed(\H^d)$ via $(g,C) \mapsto gC=\{gp:~p\in C\}$. Because $\Closed(\H^d)$ is compact, $\Prob(\Closed(\H^d))$ is compact in the weak* topology. Let 
 $\Prob_{\isom} (\Closed(\H^d)) \subset \Prob(\Closed(\H^d))$ be the sub-space of isometry-invariant measures. This space is closed and thus compact.

\begin{definition}
   Given an isometry-invariant Borel probability measure $\nu$ on $\Closed(\H^d)$, define 
   $$\density(\nu) = \nu(\Closed_\cO(\H^d))$$
   where $\Closed_\cO(\H^d)=\{C\in \Closed(\H^d):~\cO \in C\}$ (recall that $\cO\in \H^d$ is a fixed choice of origin). 
\end{definition}
For $C \subset \H^d$, let $\partial C = \bar{C} \cap \overline{\H^d \setminus C}$ be its topological boundary.

\begin{lemma}\label{L:continuity-general}
    Let $\nu$ be an isometry-invariant Borel probability measure on $\Closed(\H^d)$ and let $\mathtt{supp}(\nu) \subset \Closed(\H^d)$ be its support. Suppose for every $C \in \mathtt{supp}(\nu)$, $\vol_{\H^d}(\partial C)=0$. Then $\nu$ is a continuity point for the density functional on $\Prob_{\isom}(\Closed(\H^d))$. In other words, if $(\nu_i)_i$ is a sequence in $\Prob_{\isom}(\Closed(\H^d))$ which converges to $\nu$ in the weak* topology, then $\lim_i \density(\nu_i)=\density(\nu)$. 
\end{lemma}

\begin{proof}
    Observe that $\partial \Closed_\cO(\H^d)=\{C\in \Closed(\H^d):~\cO \in \partial C\}$. By the Portmanteau Theorem, it suffices to prove $\nu(\partial \Closed_\cO(\H^d))=0$. 

    Let $G=\Isom(\H^d)$ and $\vol_G(\cdot)$ be a Haar measure on $G$. Define $F:\Closed(\H^d) \times G \to \R$ by $F(C,g)=1$ if $g\cO \in \partial C$ and $F(C,g)=0$ otherwise. 

For a fixed closed set $C$
    $$\int F(C,g)~d\vol_G(g) = \vol_G(\{g\in G:~g\cO\in \partial C)= \vol_{\H^d}(\partial C).$$
So by hypothesis,
$$\iint F(C,g)~d\vol_G(g)d\nu(C)=\int \vol_{\H^d}(\partial C)~d\nu(C)=0.$$
On the other hand, Tonelli's Theorem implies we can switch the order of the integrals:
\begin{align*}
 0=   \iint F(C,g)~d\vol_G(g)d\nu(C)&=\iint F(C,g)~d\nu(C)d\vol_G(g) \\
    &=\iint F(gC,g)~d\nu(C)d\vol_G(g)
\end{align*}
where the last line holds because $\nu$ is $G$-invariant.  However, $F(gC,g)=1$ if and only if $\cO \in \partial C$. So for fixed $g\in G$, $\int F(gC,g)~d\nu(C) = \nu(\partial \Closed_\cO(\H^d))$. 
Thus $\nu(\partial \Closed_\cO(\H^d))=0$. 
\end{proof}

\begin{corollary}\label{cor:continuityofdensity1}
The functions $\density_r:\M_r(\H^d) \to [0,1]$ are weak* continuous.
\end{corollary}

Given $\Pi \in \cP_r(\H^d)$, let $N_r(\Pi) \in \Closed(\H^d)$ be the union of the radius-$r$ balls centered at the points of $\Pi$. The map $\Pi \mapsto N_r(\Pi)$ is continuous, injective and isometry-equivariant. Let $\Closed_r(\H^d) =\{N_r(\Pi):~\Pi \in \cP_r(\H^d)\}$ be the subspace of radius-$r$ packings. Note that the pushforward $N_{r*}$ embeds $\M_r(\H^d)$ into  $\Prob_{\isom}(\Closed(\H^d))$ and this map preserves density in the sense that
$$\density(N_{r*}(\mu)) = \density_r(\mu).$$
In future applications, we will consider $\M_r(\H^d)$ to be a subset of $\Prob_{\isom}(\Closed(\H^d))$ by identifying it with its image $N_{r*}(\M_r(\H^d))$. The reason for doing this is to make sense of weak* limits of the form $\mu = \lim_{n\to\infty} \mu_n$ where $\mu_n \in \M_{r_n}(\H^d)$ and the radii $r_n$ have a limit in $(0,\infty]$ as $n\to\infty$. The case $\lim_n r_n = \infty$ is discussed in the next subsection.

\subsection{Horoballs and their packings} \label{sec:horoballs}

Here we define horoballs, which may be thought of as balls in $\H^d$ of radius $r = \infty$.  
A proper subset $H \subset \H^d$ is a \textbf{horoball} if there exists a sequence of balls $B_n \subset \H^d$ such that the radius of $B_n$ tends to infinity as $n\to\infty$ and $B_n$ limits on $H$ in the Fell topology.

Here is a precise example of a horoball. Let
$$\mathbb{U}^d=\{x=(x_1,\ldots, x_d) \in \R^d:~ x_d>0\}$$
be the upper half-space with the Riemannian metric at $x\in \mathbb{U}^d$ given by
$$g_x(v,w) = \frac{\langle v,w \rangle}{x_d^2}$$
where we have identified the tangent space at $x$ with $\R^d$ in the usual manner. This space is isometric to hyperbolic space $\H^d$. 

Then for each $t>0$, the set 
\begin{align*}
    H_t&=\{x\in \mathbb{U}^d:~x_d \ge t\}, \quad
    \partial  H_t=\{x\in \mathbb{U}^d:~x_d = t\}
\end{align*}
is a horoball.

A \textbf{horoball packing} of $\H^d$ is a collection $P$ of horoballs with pairwise disjoint interiors. Let $\Closed_\infty(\H^d)$ denote the set of all horoball packings. Because we can recover a horoball packing from the union of its members, we may regard $\Closed_\infty(\H^d)$ as a subset of $\Closed(\H^d)$. It is not a closed subset of $\Closed(\H^d)$ because its closure contains the whole space $\H^d \in \Closed(\H^d)$ and the empty set, neither of which are  horoballs. However, the disjoint union of $\Closed_\infty(\H^d)$ with $\{\emptyset, \H^d\}$ is compact. Moreover, it is invariant under the isometry group $\Isom(\H^d)$. 

The next result is a Corollary of Lemma \ref{L:continuity-general}.
\begin{corollary}\label{cor:continuityofdensity2}
The function $\density$ from the space of invariant Borel probability measures on $\Closed_\infty(\H^d)$ to $[0,1]$ is continuous in the weak* topology.
\end{corollary}

Let $\M_\infty(\H^d) \subset \Prob_{\isom}(\Closed(\H^d))$ denote the space of all invariant Borel probability measures on $\Closed(\H^d)$ with $\mu(\Closed_\infty(\H^d))=1$. As usual, for $\mu \in \M_\infty(\H^d)$ we define its density by $\density(\mu)=\mu(\{P:~ \cO \in \cup P\})$ where $\cup P$ denotes the union of all horoballs in $P$.

Define the optimal, optimal periodic and optimal Poisson densities by: 
\begin{align*}
D_{\opt}(\H^d,\infty) &= \sup\{\density(\mu):~\mu \in \M_\infty(\H^d)\}\\
D_{\per}(\H^d,\infty) &= \sup\{\density(\mu):~\mu \in \M_\infty(\H^d), \mu \textrm{ is periodic} \}\\
D_{\Pois}(\H^d,\infty) &= \sup\{\density(\mu):~\mu \in \M_\infty(\H^d), \mu \textrm{ is a weak Poisson factor}\}.
\end{align*}
In Section \ref{S:open} we highlight a few open problems about these densities, in particular: are these functions continuous at $r=\infty$? Is $D_{\Pois}(\H^d,\infty)=0$? 

The simplex bound (discussed in \S \ref{secLatticeMeasures}) generalizes in a straightforward way to horoball packings (see \cite[Theorem 4]{boroczky1978packing}). The next result is well known; it follows from the last section of \cite{boroczky1978packing}.
\begin{theorem}\label{T:horoball1}
    There exists a unique measure $\mu \in \Prob_{\isom}(\Closed_\infty(\H^2))$ with density equal to $D_{\opt}(\H^2,\infty)$. Moreover, it is periodic and realizes the simplex bound. In particular, $D_{\opt}(\H^2,\infty)=3/\pi$.
\end{theorem}
\begin{proof}
    Let $T$ be an ideal triangle of $\H^2$. Then $T$ tiles $\H^2$ by reflection. That is, if $\G\le \Isom(\H^2)$ is the subgroup generated by reflections in the sides of $T$ then $\G$ is a non-uniform lattice and the translates $gT$ for $g\in \G$ tile the plane. 

    Let $P$ be the unique horoball packing with the property that the intersection of $P$ with any tile $gT$ (for $g\in \G$) consists of the disjoint union of three sectors which touch at the three midpoints of the sides of $gT$. This packing uniquely realizes the simplex bound (see, e.g., the remark after Proposition 2.2 in \cite{kellerhals1998ball}).  The exact value $D_{\opt}(\H^2,\infty)=3/\pi$ follows from the fact that the ideal triangle has area $\pi$ while each sector in the intersection $P \cap T$ has area $1$.
\end{proof}

\subsection{Partial continuity of optimal densities}

It will be useful for us to prove various forms of continuity for the three densities $\Dopt,D_{\Pois}$ and $D_{\per}$ {as functions of the radius $r$}.  We start with the following basic fact.
\begin{lemma}\label{L:PoissonGap2}
 Fix a dimension $d\ge 2$. Then the functions $D_{\Pois}(\H^d,\cdot), D_{\per}(\H^d,\cdot), D_{\opt}(\H^d,\cdot)$ are left-continuous on $(0,\infty)$ and upper semi-continuous on $(0,\infty]$. 
\end{lemma}

\begin{proof}

 To prove upper semi-continuity, fix a radius $\rho \in (0,\infty]$ and let $(\rho_i)_{i=1}^\infty$ be a sequence with $\lim_i \rho_i = \rho$. Let $\nu_i \in \Prob_{\isom}(\Closed_{\rho_i}(\H^d))$ be an $\Isom(\H^d)$-invariant measure on radius-$\rho_i$ packings. Because $\Prob_{\isom}(\Closed(\H^d))$ is weak* compact, there exists a sub-sequential limit $\nu_\infty\in \Prob_{\isom}(\Closed(\H^d))$ which is supported on radius-$\rho$ packings.  Because $\density$ is continuous (Lemma \ref{L:continuity-general}), the density of $\nu_\infty$ is the limit of the densities of $\nu_i$.

    Note that if each $\nu_i$ is a weak Poisson factor, then $\nu_\infty$ is also a weak Poisson factor.  Thus
    $$D_{\Pois}(\H^d,\rho) \ge \density(\nu_\infty) = \lim_i \density(\nu_i).$$
Take the sup over all measures $\nu_i$ which are weak Poisson factors, to obtain upper semi-continuity of $D_{\Pois}(\H^d,\cdot)$. The proofs for $D_{\per}(\H^d,\cdot)$ and $D_{\opt}(\H^d,\cdot)$ are similar.

To prove left-continuity, let $\mu \in \M_r(\H^d)$ be an invariant measure on $\cP_r(\H^d)$. If $t<r<\infty$ then $\cP_t(\H^d) \supset \cP_r(\H^d)$ and so we can regard $\mu$ as an element of $\M_t(\H^d)$. Moreover,  
\begin{align*}
    \density_t(\mu) &= \frac{\vol(B_t)}{\vol(B_r)}\density_r(\mu).
\end{align*}
Thus for any $t<r<\infty$
\begin{align}\label{E:continuity1}
    D_*(\H^d,t) &\ge \frac{\vol(B_t)}{\vol(B_r)}D_*(\H^d,r)
\end{align}
for all $* \in \{\Pois, \per,\opt\}$. In particular, this shows $\liminf_{t\nearrow r} D_{*}(\H^d,t) \ge D_{*}(\H^d,r)$. Because these functions are upper semi-continuous, we also have the opposite inequality. This proves left-continuity. 
\end{proof}

We will also require that $\Dopt(\H^2,r)$ is continuous.

\begin{theorem}\label{T:continuity2}
    $D_{\opt}(\H^2,r)=D_{\per}(\H^2,r)$ is continuous over $r\in (0,\infty]$. 
\end{theorem}

\begin{proof}
    The paper \cite{bowen2003periodic} proves this statement for $r \in (0,\infty)$. Theorem  \ref{T:horoball1} proves $D_{\opt}(\H^2,\infty)=D_{\per}(\H^2,\infty)$. So we only need to prove continuity at infinity. 

       Recall from Section \ref{secLatticeMeasures}, that if, for $n\ge 7$, $r_n>0$ is defined by $\cosh(2r_n) = \cot(2\pi/n)\cot(\pi/n)$ then 
    $$D_{\opt}(\H^2,r_n) = \frac{3\csc(\pi/n)-6}{n-6}$$
for any $n \in \{7,8,9,\ldots\}$. Thus 
    $$\limsup_{r\to\infty} D_{\opt}(\H^2,r)\ge \lim_{n \to \infty} \Dopt(\H^2,r_n) = 3/\pi =  D_{\opt}(\H^2,\infty) $$
where the last equality holds by Theorem \ref{T:horoball1}.  By \eqref{eqrndef} we see that \begin{equation*}
    r_n = \log n + O(1)
\end{equation*}
and so $\Vol_{\H^2}(B(\mathcal{O},r_n)) = (1 + O(1/n)) \Vol_{\H^2}(B(\mathcal{O},r_{n-1}))$ as $n \to \infty$.  In particular, for each $r \in [r_{n-1},r_n]$ we may use the centers of the optimal packing at radius $r_{n}$ to see that \begin{equation*}
    \Dopt(\H^2,r) \geq \left(1 - O(n^{-1}) \right) \Dopt(\H^2,r_n) \quad \text{ for all }r \in [r_{n-1},r_n]
\end{equation*}
 thus showing the analogous lower bound to deduce continuity at infinity.
\end{proof}

\section{A Gibbs measure at near optimal density}
\label{secDenseGibbs}

The main goal of this section is to show that we may construct Gibbs measures of near optimal density for each $r > 0$.  

\begin{theorem}\label{T:near-optimal}
For every radius $0<r<\infty$ and $\epsilon>0$, there exists $\lam_0>0$ such that if $\lam> \lam_0$ then there exists an isometry-invariant Gibbs measure $\mu \in \Gibbs(\H^2;r,\lambda)$ such that $\density(\mu)>\Dopt(\H^2,r)-\epsilon$.
\end{theorem}

We will begin with finite-volume Gibbs measures and then lift them to measures on $\H^d$.   
With this in mind, let $L\le \Isom(\H^d)$ be a lattice subgroup and $\lambda>0$. Let $\vol_M$ be the measure on $M = \H^d/L$ obtained from the volume form on $\H^d$ and $d_M$ the metric on $M$ obtained from $\H^d$.

Recall from \eqref{eq:partition-function-def}  the partition function for the finite-volume hard sphere Gibbs measure on $M$ is $$Z(L;\lambda,r) = \sum_{m \geq 0} \frac{\lambda^m}{m!} \int_{M^m} \prod_{i < j} \one\{d_M(x_i,x_j) \geq 2r\} \,d\Vol_M^{\otimes m}(\mathbf{x})$$ and from  \eqref{eq:finite-vol-gibbs-def} the finite-volume hard sphere Gibbs measure $\mu_{\lambda,r}^L \in \Prob(\cP_r(M))$ is defined by $$\mu_{\lambda,r}^L(f) = \frac{1}{Z(L;\lambda,r)}\sum_{m \geq 0} \frac{\lambda^m}{m!} \int_{M^m} \prod_{i < j} \one\{d_M(x_i,x_j) \geq 2r\} f\left(\sum_{j} \delta_{x_j}\right)\,d\Vol_M^{\otimes m}(\mathbf{x})\,.$$

Let $\tilde{\mu}^L_{\lambda,r} \in \Prob(\cP_r(\H^d))$ be the isometry-invariant lift of $\mu^L_{\lambda,r}$ guaranteed by Fact \ref{fact:lift}.  We first construct a Gibbs measure, with no claim yet on its density.

\begin{lemma}\label{lem:construct-infinite-volume-from-limits}
    Let $(L_i)_{i=1}^\infty$ be a sequence of lattice subgroups in $\Isom(\H^d)$ such that the sequence $(\H^d/L_i)_{i \geq 1}$ Benjamini-Schramm converges to $\H^d$. Then a subsequence of the measures $\tilde{\mu}^{L_i}_{\lambda,r}$ converges to an isometry-invariant Gibbs measure $\mu \in \Gibbs(\H^d;r,\lambda)$.
\end{lemma}

\begin{proof}
    Since $\Prob(\cP_r(\H^d))$ is weak* compact, $\tilde{\mu}^{L_i}_{\lambda,r}$ has a subsequence that converges to a measure $\mu \in \Prob(\cP_r(\H^d))$.  Since each $\tilde{\mu}^{L_i}_{\lambda,r}$ is isometry invariant we have that $\mu$ is as well.  We will show that $\mu$ satisfies the DLR equations \eqref{eq:DLR}.  Recall that it is sufficient to prove \eqref{eq:DLR} for each bounded measurable set $\Lambda$ and each continuous \emph{local} function $f: \cP_r(\bX) \to [0,\infty)$ where we recall that the definition of a local function $f$ is that there is an $R > 0$ sufficiently large so that  $f(\Pi)$ depends only on the restriction of $\Pi$ to the ball $B_R(\cO)$. Let $R' \geq R$ be large enough so that $\Lambda \subset B_{R'}(\cO)$.  Since $\H^d / L_i$ Benjamini-Schramm converges to $\H^d$, for each $\eps > 0$ there is an $i_0$ large enough so that for all $i \geq i_0$, the probability the injectivity radius of a random point $x \in \H^d / L_i$ is at least $R'$ is at least $1 - \eps$.  We recall that the DLR equations \eqref{eq:DLR} hold in finite volume (see, e.g., \cite[eq. (5.4)]{jansen2018gibbsian}) and so the hard sphere measures  $\mu_{\lambda,r}^{L_i}$ satisfy the DLR equations on $\H^d/ L_i$.  This shows that \begin{equation*}
        \left|\int_{\cP_r(\H^d)} f \,d\tilde{\mu}_{\lambda,r}^{L_i} - \int_{\cP_r(\H^d)} \int_{\cP_r(\Lambda)} f(\eta + \gamma_{\Lambda^c}) d\mu_{\Lambda\,|\,\gamma}^\lambda(\eta) d\tilde{\mu}_{\lambda,r}^{L_i}(\gamma) \right| \leq \eps \|f\|_\infty\,.
    \end{equation*}
    Taking $i \to \infty$ and recalling that $f$ is continuous this shows \begin{equation*}
        \left|\int_{\cP_r(\H^d)} f \,d\mu - \int_{\cP_r(\H^d)} \int_{\cP_r(\Lambda)} f(\eta + \gamma_{\Lambda^c}) d\mu_{\Lambda\,|\,\gamma}^\lambda(\eta) d\mu(\gamma) \right| \leq \eps \|f\|_\infty\,.
    \end{equation*}
    Since $\eps$ is arbitrary, this establishes \eqref{eq:DLR} for all $\Lambda$ and local continuous $f$, thus completing the proof.
\end{proof}

A simple calculation will show that in finite volume one can increase $\lambda$ to approximate the densest packing.  Notably, the required lower bound on $\lambda$ is \emph{independent of the volume} of $L \backslash \H^d.$

\begin{lemma}\label{lem:packing-to-gibbs}
    Fix $r > 0$ and $\rho > 0$.  For each $\eps > 0$ and $s < r$ there is $\lambda_0$  so that if $L$ is a lattice subgroup in $\Isom(\H^d)$ so that $L \backslash \H^d$ admits an $r$-packing of density at least $\rho$ then the hard-sphere measure (on $L \backslash \H^d$) at activity $\lambda \geq \lambda_0$ of spheres of radius $s$ has density at least $\rho (1-\eps)\Vol(B_s)/\Vol(B_r)$.
\end{lemma}
\begin{proof}
    Fix $L$ and a tuple $P=(p_1,\ldots, p_N) \in (L \backslash \H^d)^N$ which we will think of as the centers of a radius-$r$ sphere packing with density at least $\rho$. This means $d(p_i,p_j) \ge 2r$ if $i<j$ and 
    $$N \Vol(B_r)/V\ge \rho$$
    where $V$ is the volume of $L \backslash \H^d=M$. 
    
    Set $s < r$, write $\beta = r - s$ and set $c = \Vol(B_\beta)$.  Let $\Omega \subset (L \backslash \H^d)^N$ be the collection of $N$-tuples $P'=(p'_1,\ldots, p'_N)$ such that there exists a permutation $\s$ of $[N]$ with $d(p'_{\s(i)},p_i) \le \beta$ for all $i$.  Then
    \begin{align*}
        Z(L;\lambda,s) 
        &\ge \frac{\lambda^N}{N!} \int_{\Omega} \prod_{i < j} \one\{d_M(x_i,x_j) \geq 2s\} \,d\Vol_M^{\otimes N}(\mathbf{x}) = (\lambda c)^N. 
    \end{align*}

    Further, for each $m\in \mathbb{N}$, note the probability there are exactly $m$ points under the radius-$s$ hard-sphere measure at activity $\lambda$ is at most
    $$Z(L;\lambda,s)^{-1}\frac{1}{m!}\lambda^m V^m \leq Z(L;\lambda,s)^{-1}\left(\frac{e \lambda V}{m} \right)^m \leq Z(L;\lambda,s)^{-1}\left(\frac{e \lambda \Vol(B_r)  N}{\rho m} \right)^m\,. $$

    Parametrize $m = \alpha N$ and use the lower bound $Z(L;\lambda,s) 
        \ge (\lambda c)^N$ to bound the probability there are exactly $m$ points by
    $$ \left(\left(\frac{e \Vol(B_r)}{\rho} \right)^\alpha \frac{1}{c\alpha^\alpha \lambda^{1 - \alpha}} \right)^N\,.  $$
    We require now that $\alpha \leq 1-\eps/2$ and $\eps>0$ is sufficiently small so that $\frac{1}{\alpha^\alpha} \le 2$. Also we require $\lambda \geq 1$. Then this quantity is bounded above by \begin{equation}
        \left(\frac{2 e \Vol(B_r)}{c\rho}  \cdot \frac{1}{ \lambda^{\eps/2}} \right)^N \leq (\eps/2) 2^{-N}
    \end{equation} 
    where the last inequality holds for all $\lambda > \lambda_0$ where $\lambda_0$ is a function of $r,\rho,\eps$ and $c$.  This implies that the probability there are fewer than $N(1-\eps)$ points is at most $(\eps/2) N 2^{-N} \leq \eps/2$.  This shows that the density of the Gibbs measure with intensity $\lambda$ is at least \begin{equation*}(1- \eps/2)\cdot N(1 - \eps/2) \cdot \frac{\Vol(B_s)}{V} \geq (1 - \eps) \rho \frac{\Vol(B_s)}{\Vol(B_r)}\,. \qedhere \end{equation*}
\end{proof}

We introduce a few more notions of density, this time specifically for Gibbs measures:

\begin{definition}
\begin{align*}
  D_{\Gibbs,\lambda}(\H^d,r)&=\sup \{\density(\mu):\mu \in \Gibbs(\H^d;r,\lambda) \cap \M_r(\H^d)\}\\
  D_{\Gibbs,\infty}(\H^d,r) &= \liminf_{\lambda \to \infty} D_{\Gibbs,\lambda}(\H^d,r)\\
    D^+_{\Gibbs,\infty}(\H^d,r) &= \liminf_{s \nearrow r} D_{\Gibbs,\infty}(\H^d,s). 
\end{align*}    
\end{definition}

We now see that {in the limit of  large activities} Gibbs measures  can achieve the optimal periodic packing density.

\begin{corollary}\label{C:near-optimal1}
     For any dimension $d\ge 2$ and radius $0<r<\infty$,
    $D^+_{\Gibbs,\infty}(\H^d,r) \ge D_{\per}(\H^d,r)$. 
\end{corollary}
  
\begin{proof}
Let $\eps>0$. By definition of $D_{\per}(\H^d,r)$,  there exists a 
lattice $L \le \Isom(\H^d)$ and an $r$-packing of $\H^d/L$ with density at least $(1-\eps)D_{\per}(\H^d,r)$. By Fact \ref{fact:residually-finite} we may find a sequence $(L_n)_{n = 1}^\infty$ of lattices so that $\H^d / L_n$ Benjamini-Schramm converges to $\H^d$ and the optimal packing of $\H^d/L_n$ of $r$-spheres has density at least $(1-\eps)D_{\per}(\H^d,r).$  

By Lemma \ref{lem:packing-to-gibbs}, given $s<r$, there is a $\lambda_0$ (depending on $s,\eps$) so that the hard-sphere measure on $\H^d/L_n$ at activity $\lambda\ge \lambda_0$ has density at least $(1-\eps)^2D_{\per}(\H^d,r)\vol(B_s)/\vol(B_r)$. 

By Lemma \ref{lem:construct-infinite-volume-from-limits}, we may take a subsequence of the measures $\tilde{\mu}^{L_n}_{\lambda,s}$ (for $\lambda \ge \lambda_0$) that converges to an isometry invariant Gibbs measure $\mu$.  By continuity of density, we have that $\density(\mu) \geq (1 - \eps)^2D_{\per}(\H^d,r)\vol(B_s)/\vol(B_r).$ 

Thus $D_{\Gibbs,\lambda}(\H^d,s) \geq (1 - \eps)^2D_{\per}(\H^d,r)\vol(B_s)/\vol(B_r)$ for all $\lambda \ge \lambda_0$ and therefore, $$D_{\Gibbs,\infty}(\H^d,s) \ge (1 - \eps)^2D_{\per}(\H^d,r)\vol(B_s)/\vol(B_r).$$ Because $\eps>0$ and $s<r$ are arbitrary, this implies the corollary.
\end{proof}

We observe that thanks to Theorem \ref{thm:dim2-facts}, in the case of $d=2$, Corollary \ref{C:near-optimal1} shows $D^+_{\Gibbs,\infty}(\H^2,r) = D_{\per}(\H^2,r)$.  We now show that assuming continuity---as we know holds in dimension $2$ by Theorem \ref{thm:dim2-facts}---the same holds for $D_{\Gibbs,\infty}$ rather than just $D_{\Gibbs,\infty}^+$.

\begin{corollary}\label{C:near-optimal2}
     For any dimension $d\ge 2$ and radius $0<r<\infty$, if the function $D_{\per}(\H^d,\cdot)$ is continuous at $r$ then  $D_{\Gibbs,\infty}(\H^d,r) \ge D_{\per}(\H^d,r)$.
\end{corollary}

\begin{proof}
   By continuity, there exists $r' > r$ so that $D_{\per}(\H^d,r') \geq (1 - \eps)D_{\per}(\H^d,r)$ and $\Vol(B_{r'})/\Vol(B_r) \leq 1 + \eps$.  By definition of $D_{\per}$, there is a periodic packing $\nu \in \M_{r'}(\H^d)$ so that $$\density_{r'}(\nu) \geq (1 - \eps)D_{\per}(\H^d,r') \geq (1 - \eps)^2 D_{\per}(\H^d,r).$$
      In particular, there is a lattice $L$ so that $L \backslash \H^d$ has an $r'$-packing of density at least $(1 - \eps)^2\Dopt(\H^2,r)$.   
   By Fact \ref{fact:residually-finite} we may find a sequence $(L_n)_{n = 1}^\infty$ of lattices so that $L_n \backslash \H^d$ Benjamini-Schramm converges to $\H^d$ and the optimal packing of $L_n \backslash \H^d$ of $r'$-spheres has density at least $(1 - \eps)^2D_{\opt}(\H^d,r).$  By Lemma \ref{lem:packing-to-gibbs} this implies that we may take $\lambda$ large enough so that the measures $\tilde{\mu}^{L_n}_{\lambda,r}$ for $n$ large enough have density at least 
   $$(1 - \eps)^3 D_{\opt}(\H^d,r)\frac{\Vol(B_r)}{\Vol(B_{r'})} \geq  \frac{(1 - \eps)^3}{1+\eps} D_{\opt}(\H^d,r) \geq (1 - \eps)^4D_{\opt}(\H^d,r).$$ 
   By Lemma \ref{lem:construct-infinite-volume-from-limits}, we may take a subsequence of the measures $\tilde{\mu}^{L_n}_{\lambda,r}$ that converges to an isometry invariant $(\lambda,r)$-Gibbs measure $\mu$.
   
   By continuity of density (Lemma \ref{L:continuity-general}), we have that $\density(\mu) \geq (1 - \eps)^4D_{\per}(\H^d,r).$ Thus $D_{\Gibbs,\lambda}(\H^d,r) \ge (1 - \eps)^4D_{\per}(\H^d,r)$. Since $\lambda \ge \lambda_0$ is arbitrary, this implies $D_{\Gibbs,\infty}(\H^d,r) \ge (1 - \eps)^4D_{\per}(\H^d,r)$. Since $\eps>0$ is arbitrary, this implies the lemma.
\end{proof}

\begin{proof}[Proof of Theorem \ref{T:near-optimal}]
This follows immediately from Corollary \ref{C:near-optimal2}, Theorem \ref{T:continuity2} (which shows $D_{\opt}(\H^2,r)$ is continuous) and Theorem \ref{thm:dim2-facts} (which shows that $D_{\opt}(\H^2,r)=D_{\per}(\H^2,r)$).
\end{proof}

\section{A Gibbs measure of lower density} \label{sec:lower-density}

The main goal of this section is to show the existence of a Gibbs measure that is a weak Poisson factor.  Our proof is quite general, and so we will work in more generality than $\H^d$.

\begin{assumption}\label{assumption-on-X}
	Let $G$ be a locally compact second countable group, $K \leq G$ a compact subgroup and set $\bX = G / K$.  Let $\Vol_\bX$ be a $G$-invariant Radon measure on $\bX$ and let $d$ be a left-invariant proper metric on $\bX$.  We assume that there is a sequence of uniform lattices $L_n \leq G$ so that the sequence $M_n = L_n \backslash \bX$ Benjamini-Schramm converges to $\bX$.  We also assume that spheres in $\bX$ have measure zero. 
\end{assumption}

\begin{theorem}\label{th:sparser-measure}
	Let $\bX$ satisfy Assumption \ref{assumption-on-X}.  For each $r > 0$ and $\lambda \geq 0$ there is an $G$-invariant Gibbs measure $\mu \in \Gibbs(\bX;r,\lambda)$ that is a weak Poisson factor.  In particular, $\density_r(\mu) \leq D_{\Pois}(\bX,r)$.
\end{theorem}

We will construct a sequence of probability measures $(\nu_t)_{t \geq 0} \subset \M_{r}(\bX)$ via \emph{spatial birth-death Glauber dynamics} starting from the empty point set $\nu_0 = \delta_\emptyset$.    
The dynamics governing the random process $(\eta_t)_{t \geq 0} \subset \cP_{r}(\bX)$ can be informally described as follows.
Points arrive on $\bX$ at rate $\lambda$ tagged with an independent standard exponential random variable $\tau$.  Each time a point arrives at $x$ at some time $t$ tagged with $\tau$, we check if $\eta_{t^-}(B_{2r}(x)) := \lim_{s \to t^-} \eta_s(B_{2r}(x)) = 0$; if so, then we add the point $x$ to $\eta_t$, and later remove it from $\eta_{t + \tau}$.  The measures $(\nu_t)_{t \geq 0} \subset \M_{r}(\bX)$ are the laws of $(\eta_t)_{t \geq 0}.$  See Figure \ref{fig:glauber} for a simulation at large activity.

\begin{figure}[htbp]
	\centering
	\includegraphics[width=0.32\linewidth]{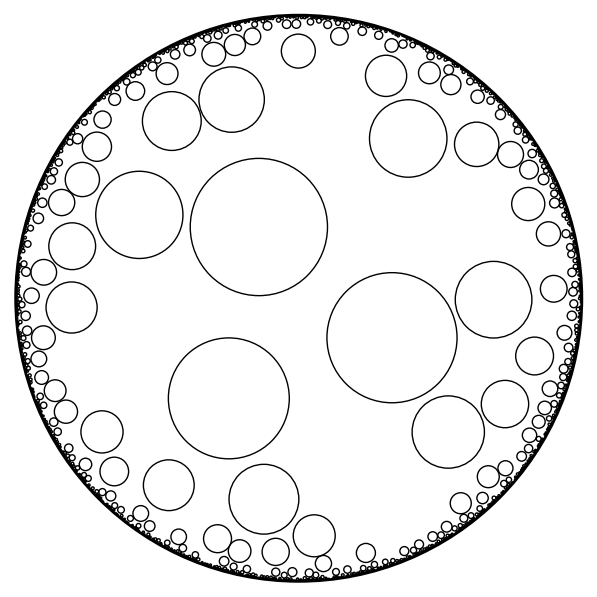}
	\caption{A sample from spatial birth-death Glauber dynamics with radius $r_7$, activity $\lambda = 10000$ and $T = 40$.}
	\label{fig:glauber}
\end{figure}

Spatial birth-death dynamics for the hard sphere model---and Gibbs point processes more generally---are well-studied from various different perspectives, beginning with the foundational works of Holley and Stroock \cite{holley1978nearest} which constructed nearest-neighbor birth and death processes on the real line, taking inspiration from Spitzer's work \cite{spitzer1977stochastic} on similar processes on the integers.   Since then, there have been various works constructing and analyzing spatial birth-death dynamics more generally \cite{finkelshtein2012correlation,finkelshtein2015statistical,kondratiev2005glauber,finkelshtein2012semigroup,garcia1995birth,garcia2006spatial}.  This list is by no means exhaustive, and we refer the reader to the references therein for more background and context.  While many works take a semi-group approach to the construction, we explicitly write our dynamics as a thinned Poisson process, similar to Garcia's work on birth and death processes \cite{garcia1995birth}.   We will describe the dynamics more precisely and prove relevant properties in Section \ref{sec:GD-construction}.

A related argument appears in the 2025 preprint of Jahnel-K\"oppl-Steenbeck-Zass \cite{jahnel2025reversible} which constructs Glauber dynamics on $\R^d$ via a Poisson thinning.  Further, the work \cite{jahnel2025reversible} shows that if one starts Glauber dynamics from a sufficiently regular measure, then all subsequential time limits are Gibbs measures.  Their entropy dissipation approach follows works of Holley and Stroock \cite{holley1978nearest} and Holley \cite{holley1971free} from the discrete setting.  A main challenge for us is to apply these ideas in the non-amenable setting.  In this direction, Shriver \cite{shriver2023free} recently showed an analogous non-amenable analogue of Holley's work \cite{holley1971free} in the discrete setting by working on a sequence of quotients that Benjamini-Schramm converge to the underlying space.  Our approach is most directly inspired by Shriver's work paired with the graphical construction of \cite{garcia1995birth}.

To simplify our notation, we will write $\Omega = \cP_r(\bX)$.   A \textbf{good test function} $F: \bX \times \Omega \to \R$ is a bounded continuous function so that there is a compact set $K \subset \bX$ and a finite $R > 0$ so that $F(x,\eta) = 0$ for $x \notin K$ and so that $F(x,\eta)$ depends only on $(x,\eta|_{B(x,R)})$.  In particular, $F$ has compact support in the first coordinate and is local in the second coordinate.  Since $\Omega$ is compact and metrizable, there is a countable dense collection of good test functions.  As such, in order to prove that a measure $\nu$ is a Gibbs measure, it will be sufficient to prove that it satisfies the GNZ equations for some such collection of good test functions.

For $\nu \in \Prob(\Omega)$ and a good test function $F$ define \begin{align*}
    B_\nu(F) &= \lambda \int_{\bX} \int_{\Omega} F(x,\eta) e^{-H(x\,|\,\eta)} \,d\nu(\eta)\,d\Vol_{\bX}(x) \\
    D_\nu(F) &=  \int_{\Omega} \sum_{x \in \eta} F(x,\eta - \delta_x) \,d\nu(\eta)\,.
\end{align*}

Recall that by the GNZ equations \eqref{eq:GNZ-equations}, if $B_\nu(F) = D_\nu(F)$ for all bounded $F$ with compact support in the $\bX$ coordinate, then $\nu$ is a Gibbs measure.  

With this in mind, define the \textbf{GNZ defect} \begin{equation*}
    \Delta_\nu(F) = D_\nu(F) - B_\nu(F)\,.
\end{equation*}

Our goal is to show that $\Delta_\nu(F) = 0$ for all good test functions.  We first describe how to define Glauber dynamics in both finite and infinite volume.  

\subsection{Glauber dynamics from an empty start}\label{sec:GD-construction}

We provide a graphical construction of Glauber dynamics from an empty start.  We will work in the infinite-volume setting but will also use this construction for the finite-volume setting.

Consider the space $\mathcal{X} = \bX \times \R_+ \times \R_+$ and the Poisson process $\mathcal{N}$ on $\mathcal{X}$ with intensity measure $\lambda e^{-\ell} d\Vol_{\bX} \otimes dt \otimes d\ell$, where $dt$ and $d\ell$ are standard Lebesgue measure.  We define a relation $\prec$ on $\mathcal{X}$ by $$(x',t',\ell') \prec (x,t,\ell) \iff t' < t \text{ and } d(x',x) < 2r\,.$$

For each $T < \infty$, we describe the rule for constructing the measures $(\nu_t)_{t \in [0,T]} \subset \M_{r}(\bX)$.  We work on the space $\mathcal{X}_T = \bX \times [0,T] \times \R_+$.  Given an instance of the process $\mathcal{N}$  and a point $(x,t,\ell) \in \mathcal{N}$ define the \emph{influence cluster} $I_{\mathcal{N}}(x,t,\ell)$  via: 
\begin{align*}
	I_{\mathcal{N}}(x,t,\ell) =  \{(x_k,t_k,\ell_k)& \in \mathcal{N} : \exists~k \in \bN,\exists~\{(x_j,t_j,\ell_j)\}_{j = 0}^k \in \mathcal{N} \\
	&\text{ with } (x_0,t_0,\ell_0) = (x,t,\ell),  (x_j,t_j,\ell_j) \prec (x_{j-1},t_{j-1},\ell_{j-1}) \text{ for all }j \in [k]\}\,. 
\end{align*}

We will show that almost-surely every $I_{\mathcal{N}}(x,t,\ell)$ is finite. 
\begin{lemma}\label{lem:finite-influence-cluster}
	Almost-surely, for all $(x,t,\ell) \in \mathcal{N}$ we have $|I_{\mathcal{N}}(x,t,\ell)| < \infty$.  
\end{lemma}
\begin{proof}
	First note that almost surely for all $(x,t,\ell) \in \mathcal{N}$ we have $\{(x',t',\ell') \in \mathcal{N} : (x',t',\ell') \prec (x,t,\ell) \}$ is finite.  By the Mecke formula, note that the expected number of points $(x',t',\ell') \in \mathcal{N} \cap B_1(\cO) \times [0,T] \times \R_+$ with $(x',t',\ell') \prec (x,t,\ell)$ is bounded above by $T \lambda \Vol_\bX(B_1)$.  Thus, for each $T$, we almost-surely have that $|\{(x',t',\ell') \in \mathcal{N} : (x',t',\ell') \prec (x,t,\ell) \}| < \infty$ for all points $(x,t,\ell) \in \mathcal{N} \cap B_1(\cO) \times [0,T] \times \R_+$.  Union bounding over countably many $T$ and countably many centers (using that $\bX$ is separable since $G$ is second countable) shows that almost surely for all $(x,t,\ell) \in \mathcal{N}$ we have $\{(x',t',\ell') \in \mathcal{N} : (x',t',\ell') \prec (x,t,\ell) \}$ is finite.
	
	Thus, if we have that $|I_{\mathcal{N}}(x,t,\ell)| = \infty$ we must have an infinite sequence $ \{(x_j,t_j,\ell_j)\}_{j \geq 0} \in \mathcal{N}$ with $(x,t,\ell) = (x_0,t_0,\ell_0)$  so that $(x_j,t_j,\ell_j) \succ (x_{j+1},t_{j+1},\ell_{j+1})$ for all $j$.  Let $\mathcal{E}_k(x,t,\ell)$ denote the event that there is such a sequence $\{(x_j,t_j,\ell_j)\}_{j  = 0}^k\,.$  It is enough to show that there is almost-surely no point in $B_1(\cO)$ for which $\mathcal{E}_k(x,t,\ell)$ holds for all $k \geq 1$. By the Mecke formula \cite[Theorem 4.4]{last2017lectures} we have  $$\EE\left[  \sum_{(x,t,\ell) \in \mathcal{N} \cap B_1(\cO) \times [0,T] \times \R_+} \one\{\mathcal{E}_k(x,t,\ell)\}  \right] = \lambda \Vol(B_1(\cO)) \int_0^\infty \int_0^T \mathbb{P}_{\cO,t,\ell}(\mathcal{E}_k(\cO,t,\ell)) \,dt \, e^{-\ell} d\ell$$
	where $\mathbb{P}_{\cO,t,\ell}$ denotes the probability over the Poisson process conditioned on having a point at $(\cO,t,\ell)$.  
	To upper bound the probability of $\mathcal{E}_k(\cO,t,\ell)$, note that on the event $\mathcal{E}_k(\cO,t,\ell)$ there is a tuple of points $(x_1,\ldots,x_k)$ so that $x_j \in B_{2r}(x_{j-1})$ (where we interpret $x_0 = \cO$) and $(x_j,t_j,\ell_j) \in \mathcal{N}$ for some $t_j \in [0,t]$, $\ell_j\in [0,\infty)$.  By the Mecke formula, the expected number of such tuples is bounded by $(\lambda \Vol_{\bX}(B_{2r}(\cO)t)^k$.  For each such tuple, note that for uniform and independent random variables $t_1,\ldots,t_k \in [0,t]$ we have that the probability $t_1 > t_2 > \cdots  > t_k$ is $1/k!.$  Thus
	\begin{equation*}
		\mathbb{P}_{\cO,t,\ell}(\mathcal{E}_k(\cO,t,\ell)) \leq \frac{1}{k!}(\lambda \Vol_\bX(B_{2r}(\cO))t)^k \xrightarrow{k \to \infty} 0 \,.
	\end{equation*}

	In particular, this implies $\bigcap_{k \geq 1} \mathcal{E}_k(x,t,\ell)$ does not hold for any $(x,t,\ell) \in \mathcal{N}$ (with probability 1), completing the proof.
\end{proof}

As an immediate consequence of the proof of Lemma \ref{lem:finite-influence-cluster} we can control the speed of propagation of information.  In particular, for a compact set $K \subset \bX$ and a time $T$ let $N_k(K,T)$ denote the number of chains $\{(x_j,t_j,\ell_j)\}_{j = 0}^k \in \mathcal{N}$ with $x_0 \in K$ and $t_0 \leq T$ so that $(x_j,t_j,\ell_j) \succ (x_{j+1},t_{j+1},\ell_{j+1})$.  We record the following corollary of the proof of Lemma \ref{lem:finite-influence-cluster}: \begin{corollary}\label{cor:finite-propagation-backwards}
	We have \begin{equation*}
		\EE N_k(K,T) \leq \lambda^{k+1} \Vol_\bX(K) \Vol_{\bX}(B_{2r}(\mathcal{O}))^k \frac{T^{k+1}}{(k+1)!}\,.
	\end{equation*}
\end{corollary}

We now describe the acceptance rule $A_{\mathcal{N}}:\mathcal{N} \to \{0,1\}$ recursively: \begin{equation}\label{eq:acceptance-rule}
	A_{\mathcal{N}}(x,t,\ell) = \prod_{\substack{(x',t',\ell') \in \mathcal{N}}} (1 - A_{\mathcal{N}}(x',t',\ell')\one\{(x',t',\ell') \prec (x,t,\ell) \wedge t\in[t',t'+\ell']\} )\,.
\end{equation}

The next lemma is a straightforward exercise giving an equivalent characterization.
\begin{lemma}
	The acceptance rule is uniquely determined by the following properties.
	\begin{enumerate}
		\item If $(x,t,\ell)\in \mathcal{N}$ is minimal in the sense that there does not exist $(x',t',\ell')\in \mathcal{N}$ with $(x',t',\ell') \prec (x,t,\ell)$, then $A_{\mathcal{N}}(x,t,\ell) =1$.
		\item If $A_{\mathcal{N}}(x_1,t_1,\ell_1)=1$, $(x_1,t_1,\ell_1) \prec (x_2,t_2,\ell_2)$ and $t_2 \in [t_1,t_1+\ell_1]$, then $A_{\mathcal{N}}(x_2,t_2,\ell_2) =0$.
		\item If $(x,t,\ell)\in \mathcal{N}$ and there does not exist $(x',t',\ell')$ satisfying $A_{\mathcal{N}}(x',t',\ell')=1$, $(x',t',\ell') \prec (x,t,\ell)$ and $t \in [t',t'+\ell']$ then $A_{\mathcal{N}}(x,t,\ell) =1$.
	\end{enumerate}
\end{lemma}

The random counting measures $\eta_s, \eta_{s^-} \in \cP_{r}(\bX)$ are defined via \begin{eqnarray}
	\eta_s &= \sum_{(x,t,\ell) \in \mathcal{N}} \delta_x \cdot A_{\mathcal{N}}(x,t,\ell)  \cdot \one\{ s \in [t,t+\ell] \}\label{eq:nu-s-formal-def}\\
	\eta_{s^-} &= \sum_{(x,t,\ell) \in \mathcal{N}} \delta_x \cdot A_{\mathcal{N}}(x,t,\ell)  \cdot \one\{ s \in (t,t+\ell] \} \nonumber
\end{eqnarray}
and $\nu_s \in \M_r(\bX)$ is the law of $\eta_s$. Note $\eta_{s^-}$ is the limit of $\eta_t$ as $t$ approaches $s$ from below.

\begin{lemma}\label{lem:mu_t-Poisson-factor}
	For each $t  > 0$, the measure $\nu_t \in \M_r(\bX)$ is a Poisson factor.   Its factorial moment densities exist and obey the Ruelle bound $\rho_{\nu_t}^{(k)} \leq \lambda^k$.
\end{lemma}
\begin{proof}
	To show that $\nu_t$ is a Poisson factor we need only show that the map $\mathcal{N} \mapsto \eta_t$ is measurable.  By Lemma \ref{lem:finite-influence-cluster}, we have that almost surely the influence cluster $I_{\mathcal{N}}(x,s,\ell)$ is finite for each $(x,s,\ell) \in \mathcal{N}$.  This assures that the acceptance rule \eqref{eq:acceptance-rule} is in fact a finite product, and so the definition of $\eta_s$ at \eqref{eq:nu-s-formal-def} is a Borel function of $\mathcal{N}$.  
	
	To see the Ruelle bound, we note that the pointset at time $t$ is a subset of $$\{x : (x,s,\ell)  \in \mathcal{N}, 0 \leq s \leq t < s + \ell\}$$
	which is precisely a Poisson process of intensity $$\lambda \int_0^t \int_{t-s}^\infty e^{-\ell} \,d\ell\,ds = \lambda(1 - e^{-t}) \leq \lambda\,.$$
    The factorial moment densities exist because the process of accepted points is a thinning of the Poisson process of intensity $\lambda$.
\end{proof}

\subsection{Finite volume}

Throughout this section, let $M = L \backslash \mathbb{X}$ be a compact quotient.   We will always have $\lambda \geq 0$ and $r > 0$ fixed. We write $\Omega_M= \cP_r(M)$ and define $\mu_M^\lambda$ to be the hard sphere measure on $M$ defined in \eqref{eq:finite-vol-gibbs-def}.

We first recall that $\mu_M^\lambda$ satisfies the GNZ equations \eqref{eq:GNZ-equations}:
For every non-negative measurable $F: M \times \Omega_M \to [0,\infty]$ we have \begin{equation}\label{eq:finite-vol-GNZ}
    \int_{\Omega_M} \sum_{x \in \eta} F(x,\eta - \delta_x) \,d\mu_M^{\lambda}(\eta) = \lambda \int_M \int_{\Omega_M} F(x,\eta) e^{-H(x\,|\,\eta)} \,d\mu_M^{\lambda}(\eta)\,d\Vol_M(x)\,.
\end{equation}

We define $\nu_0^M \in \Prob(\Omega_M)$ to be the probability measure on the empty configuration and $\nu_t^M \in \Prob(\Omega_M)$ to be the probability measure obtained by running Glauber dynamics up to time $t$.  We note that the generator of the Glauber dynamics on $M$ may be defined via \begin{equation*}
	(L_M f)(\eta) = \sum_{x \in \eta}(f(\eta - \delta_x) - f(\eta)) + \lambda \int_M e^{-H(x\,|\,\eta)}(f(\eta+\delta_x) - f(\eta))\,d\Vol_M(x)
\end{equation*}
and so we may write $\nu_t^M = \nu_0^M e^{t L_M} $, which will be useful for taking time derivatives.

It will be convenient to define the measure $Q_M$ on $M \times \Omega_M$ by \begin{equation}
    dQ_M(x,\eta) = \lambda e^{-H(x\,|\,\eta)} d\Vol_M(x) d\mu_M^{\lambda}(\eta)
\end{equation}
which measures average birth.

For a non-negative measurable $F: M \times \cP_r(M) \to [0,\infty]$ define the birth and death functionals 
\begin{align*}
    B_\nu^M(F) &= \int_M \int_{\Omega_M} \lambda e^{-H(x\,|\,\eta)} F(x,\eta) \,d\nu(\eta)\,d\Vol_M(x)\\
    D_\nu^M(F) &= \int_{\Omega_M} \sum_{x \in \eta} F(x,\eta - \delta_x) \,d\nu(\eta)\,.
\end{align*}

We will ultimately compare these functions to their infinite-volume analogues, but first note useful integral representations of them.
\begin{lemma}\label{lem:birth-death-identities}
    Suppose that $\nu \in \Prob(\Omega_M)$ is absolutely continuous with respect to $\mu_M^{\lambda}$ with $\frac{d\nu}{d\mu_M^{\lambda}} = \phi$.  Then \begin{equation*}
        B_\nu^M(F) = \iint_{M\times \Omega_M}  F(x,\eta) \phi(\eta) \,dQ_M(x,\eta)\quad \text{ and }\quad D_\nu^M(F) = \iint_{M\times \Omega_M} F(x,\eta) \phi(\eta + \delta_x) \,dQ_M(x,\eta)\,.
    \end{equation*}
\end{lemma}
\begin{proof}
    The representation for $B_\nu^M(F)$ is immediate. For the representation of $D_\nu$, note that by \eqref{eq:finite-vol-GNZ} we have \begin{align*}
        D_\nu^M(F) &= \int_{\Omega_M} \sum_{x \in \eta} F(x,\eta - \delta_x)  \phi(\eta) \,d\mu_M^\lambda(\eta) \\
        &= \iint_{M\times \Omega_M}\lambda  e^{-H(x\,|\,\eta)} F(x,\eta) \phi(\eta + \delta_x)  \,d\mu_M^\lambda(\eta) \,d\Vol_M(x) \\
        &= \iint_{M\times \Omega_M} F(x,\eta) \phi(\eta + \delta_x) \,dQ_M(x,\eta)\,. \qedhere
    \end{align*}
\end{proof}

We prove a standard integration by parts identity for the generator $L$.
\begin{fact}\label{fact:IBP-generator}
  	Let $f,g: \Omega_M\to \R$ be bounded and measurable.  Then $$\int_{\Omega_M} (L_M f)(\eta) g(\eta) \,d\mu_M^\lambda(\eta) = -\iint_{M\times \Omega_M } (f(\eta + \delta_x) - f(\eta))(g(\eta + \delta_x) - g(\eta)) \,dQ_M(x,\eta)\,.$$
\end{fact}
\begin{proof} Expand
	\begin{align*}
		\int_{\Omega_M} (L_M f)(\eta) g(\eta) \,d\mu_M^\lambda(\eta) &= \lambda \int_{\Omega_M} \int_{M} e^{-H(x\,|\,\eta)} (f(\eta + \delta_x) - f(\eta)) g(\eta)\,d\Vol_M(x) \,d\mu_M^\lambda(\eta) \\
		&\qquad + \int_{\Omega_M}\sum_{x \in \eta} (f(\eta - \delta_x) - f(\eta))  g(\eta)  \,d\mu_M^\lambda(\eta)\,.
	\end{align*}
	Applying the GNZ equation \eqref{eq:finite-vol-GNZ} to the second term we see that \begin{align*}
		 \int_{\Omega_M}\sum_{x \in \eta}& (f(\eta - \delta_x) - f(\eta))  g(\eta)  \,d\mu_M^\lambda(\eta) \\
		 &= \lambda \int_{\Omega_M} \int_{M} e^{-H(x\,|\,\eta)} (f(\eta) - f(\eta+\delta_x)) g(\eta + \delta_x)\,d\Vol_M(x) \,d\mu_M^\lambda(\eta)\,.
	\end{align*}
	Combining terms completes the proof.
\end{proof}

We note that $\nu_t^M$ is in fact absolutely continuous with respect to $\mu_M^\lambda$:

\begin{fact}\label{fact:nu_t-density}
	For every $t \geq 0$ we have $\nu_t^M \ll \mu_M^\lambda$.
\end{fact}
\begin{proof}
	We note that $\nu_0^M \ll \mu_M^\lambda$.  Since the dynamics $P_t^M = e^{t L_M}$ giving Glauber dynamics on $M$ up to time $t$ preserves $\mu_M^\lambda$, we see that for any set $A$ with $\mu_M^\lambda(A) = 0$ we have \begin{equation*}0 = \mu_M^\lambda(A) = \int_{\Omega_M} P_t^M \one_A(\eta) \,d\mu_M^\lambda(\eta) =  \int_{\Omega_M} P_t^M \one_A(\eta) \,d\nu_M^0(\eta) = \nu_M^t(A) \end{equation*}
    where the penultimate equality follows from $\nu_M^0 \ll \mu_M^\lambda.$
\end{proof}

We now define the entropy and information with respect to the measure $\mu_M^\lambda$.

\begin{definition}[Entropy and information]\label{def:entropy-and-information}
    Suppose that $\nu \in \Prob(\Omega_M)$ is absolutely continuous with respect to $\mu_M^{\lambda}$ with $\frac{d\nu}{d\mu_M^{\lambda}} = \phi$.  Define \begin{align*}
        \mathrm{Ent}_{\mu_M^{\lambda}}(\nu) &= \int_{\Omega_M} \phi(\eta) \log \phi(\eta) \, d\mu_M^{\lambda}(\eta) \\
        I_M(\nu) &=  \int_M \int_{\Omega_M} \left(\phi(\eta + \delta_x) - \phi(\eta) \right)\left[\log \phi(\eta + \delta_x) - \log \phi(\eta) \right] \,dQ_M(x,\eta)
    \end{align*}
    where we use the convention that $(x - y)(\log x - \log y) = 0$ if $x = y = 0$ and $(x - y)(\log x - \log y) = +\infty$ if only one is $0$.
\end{definition}

Our main use of $I_M$ is to bound the difference between $B_\nu^M(F)$ and $D_\nu^M(F)$.  We will later control $I_M$ when averaged over time. \begin{lemma}\label{lem:finite-volume-difference}
    Let $\nu \in \Prob(\Omega_M)$ be absolutely continuous with respect to $\mu_M^\lambda$.  For every bounded and measurable $F: M \times \Omega_M \to \R$ we have \begin{align*}
        |B_\nu^M(F) - D_\nu^M(F)|^2 &\leq \left(B_\nu^M(F^2) + D_\nu^M(F^2) \right) \iint_{M \times \Omega_M} \frac{(\phi(\eta) - \phi(\eta + \delta_x))^2}{\phi(\eta) + \phi(\eta + \delta_x)} \,dQ_M(x,\eta) \\
        &\leq  \left(B_\nu^M(F^2) + D_\nu^M(F^2) \right) I_M(\nu)\,.
    \end{align*}
\end{lemma}
\begin{proof}
    By Lemma \ref{lem:birth-death-identities} along with Cauchy-Schwarz we bound \begin{align*}
        |B_\nu^M(F) - D_\nu^M(F)|^2 &= \left(\iint_{M \times \Omega_M} F(x,\eta)[\phi(\eta) - \phi(\eta + \delta_x)] \,dQ_M(x,\eta)\right)^2 \\
        &\leq \left(\iint_{M \times \Omega_M} F(x,\eta)^2[\phi(\eta) + \phi(\eta + \delta_x)] \,dQ_M(x,\eta)\right) \\
        &\qquad \times\left(\iint_{M \times \Omega_M} \frac{(\phi(\eta) - \phi(\eta + \delta_x))^2}{\phi(\eta) + \phi(\eta + \delta_x)} \,dQ_M(x,\eta)\right)\,. 
    \end{align*}
    The first factor is precisely $B_\nu^M(F^2) + D_\nu^M(F^2)$ by Lemma \ref{lem:birth-death-identities}.  For the second factor, note that $\frac{(x - y)^2}{x + y} \leq (x - y)(\log x - \log y)$ and use Definition \ref{def:entropy-and-information}.
\end{proof}

The fundamental theorem of calculus will show us that $I_M$ can be bound on average by the volume of $M$.

\begin{lemma}\label{lem:FTC-entropy}
    Suppose that $\nu_0^M \in \Prob(\Omega_M)$ is the point mass on the empty configuration and set $\nu_t^M$ to be the dynamics run up to time $t$ on $M$.  Then $$\int_0^T I_M(\nu_t) \,dt \leq \mathrm{Ent}_{\mu_M^\lambda}(\nu_0) = \log Z_M(\lambda) \leq \lambda \Vol_M(M)\,.$$
\end{lemma}
\begin{proof}
	Let $\phi_t = \frac{d \nu_t^M}{d\mu_M^\lambda}$ which is guaranteed to exist by Fact \ref{fact:nu_t-density}.  For $\delta > 0$ define $\phi_t^\delta = (1 - \delta)\phi_t + \delta$ which is precisely the density of the measure obtained when starting at $(1 - \delta)\nu_0^M + \delta \mu_M^\lambda$ and evolving for time $t$.  Note that for each fixed $\delta$ we have $\phi_t^\delta \geq \delta$ and that $\phi_t^\delta$ is uniformly upper bounded in $t$ since $\phi_t = e^{tL_M} \phi_0$ and $\|e^{tL_M} \phi_0\|_\infty \leq \|\phi_0\|_\infty$.  We may differentiate under the integral to see that
    \begin{align*}
		\frac{d}{dt} \mathrm{Ent}_{\mu_M^\lambda}(\phi_t^\delta \mu_M^\lambda) &= \int_{\Omega_M} \frac{d}{dt}\left(\phi_t^\delta(\eta) \log \phi_t^\delta(\eta) \right) \,d\mu_M^\lambda(\eta) = \int_{\Omega_M} \left(\frac{d}{dt}\phi_t^\delta(\eta)\right) \left(\log \phi_t^\delta(\eta) + 1 \right) \,d\mu_M^\lambda(\eta) \\
		 &= \int_{\Omega_M} \left(\frac{d}{dt}\phi_t^\delta(\eta)\right) \log \phi_t^\delta(\eta) \,d\mu_M^\lambda(\eta) \\
		 &= \int_{\Omega_M} \left(L\phi_t^\delta(\eta)\right) \log \phi_t^\delta(\eta) \,d\mu_M^\lambda(\eta) \\
		 &= - I_M(\phi_t^\delta)
	\end{align*}
	where the last equation is by Fact \ref{fact:IBP-generator}.  By the fundamental theorem of calculus we have \begin{equation*}
		\int_0^T I_M(\phi_t^\delta \mu_M^\lambda)\,dt \leq \mathrm{Ent}_{\mu_M^\lambda}(\phi_0^\delta)\,.
	\end{equation*}
	Taking $\delta \to 0$ we see that the right-hand side converges to 	$\mathrm{Ent}_{\mu_M^\lambda}(\nu_0^M)$ by the dominated convergence theorem.  The integrand on the left-hand side converges pointwise and so Fatou's lemma completes the proof.
\end{proof}

Our use of Lemma \ref{lem:FTC-entropy} will be the following bound on $I_M(\nu_t)$ at a certain time in terms of $\Vol_M(M)$.  \begin{corollary}\label{cor:bound-on-information}
    For each $T > 0$ and compact $M = L \backslash \bX$ there is some time $\tau \in [T/2,T]$ so that $$I_M(\nu_\tau^M) \leq \frac{2\lambda \Vol_M(M)}{T}\,.$$
\end{corollary}

\subsection{Lifting to infinite-volume}

We now want to show that if we lift the measures $\nu_t^M$ to a measure on $\Omega(\bX)$ then we are close to a Gibbs measure.  While one could work with the lifted versions guaranteed by Fact \ref{fact:lift}, we instead directly lift the birth and death functionals $B$ and $D$. We pause for a moment to properly define our lifting operations.

For a compact $M = L \backslash \bX$ let $\pi : \bX \to M$ be the standard quotient.  Given $g \in L \backslash G$ define the projection $\pi_g$ via $\pi_g(x) = \pi(gx)$.  For a pointset $\eta \in \Omega$, we may define the periodic lift via \begin{equation*}
    \theta_g \eta = \sum_{x \in \eta}\sum_{y \in \pi_g^{-1}(x)}\delta_y\,.
\end{equation*}
We note that in general we may have $\theta_g \eta \notin \cP_r(\bX)$ for points where the injectivity radius is small.  Since we will have that most points of $M$ will have large injectivity radius by Assumption \ref{assumption-on-X}, we will be able to discard such points.  In this direction, for a good test function $F: \bX \times \Omega \to \R$ let $K \subset \bX$ be the compact support of the first coordinate and set $R \geq 2r$ so that $F(x,\eta)$ depends only on $\eta|_{B_{R}(x)}$.  For $g \in L \backslash G$ define $$K_g^\circ = \{x \in  K : \pi_g \text{ is injective on }B_R(x)\}$$
and for $x\in M$ define
$$Z_g(x) = \pi_g^{-1}(x) \cap K_g^\circ\,.$$
The motivation for this definition is that when $x \in K_g^\circ$, the restriction of $\theta_g \eta$ to $B_R(x)$ in fact lies in $\cP_r(\bX)$.  In particular, this implies that \begin{equation}\label{eq:Z-bound}
	|Z_{g}(x)| \leq O_{K,R}(1)
\end{equation}
for a constant depending only on $K,R$.  For $(x,\eta)\in M\times \Omega_M$, define \begin{equation*}
	F_g^M(x,\eta) := \sum_{z \in Z_g(x)} F(z,\theta_g \eta)\,.
\end{equation*}

Define $\P_{L \backslash G}$ to be the Haar measure on $L \backslash G$ normalized to be a probability measure.  We then define the lifted versions of $B$ and $D$ via \begin{align*}
	\widetilde{D}_\nu^M(F) &= \int_{L \backslash G} D_\nu^M(F_g^M) \, d\P_{L\backslash G}(g) \\
	\widetilde{B}_\nu^M(F) &= \int_{L \backslash G} B_\nu^M(F_g^M) \, d\P_{L\backslash G}(g)
\end{align*}
and define the finite-volume defect $$\Delta_\nu^M(F) = \widetilde{D}_\nu^M(F) - \widetilde{B}_\nu^M(F)\,.$$

We will later see that the finite propagation bound in Corollary \ref{cor:finite-propagation-backwards} will allow us to compare $\Delta_{\nu_t^M}^M(F)$ to $\Delta_{\nu_t}(F)$.  We first bound $\Delta_\nu^M(F)$.

\begin{lemma}\label{lem:finite-volume-dissipation}
    Let $F: \bX \times \Omega \to \R$ be a good test function.  Then there is a constant $C(F,r) > 0$ so that \begin{equation*}
        |\Delta_{\nu_t^M}^M(F)|^2 \leq  \lambda C(F,r) \frac{I_M(\nu_t^M)}{\Vol_M(M)} \,.
    \end{equation*}
\end{lemma}
\begin{proof}
	Write $\nu = \nu_t^M$ and $\phi = \frac{d\nu}{d \mu_M^\lambda}$ for simplicity.  Define $E_g = \{(x,\eta) : Z_g(x) \neq \emptyset\} \subset M\times \Omega_M$.  By definition of $F_g^M(x,\eta)$ we have that $F_{g}^M$ has support contained in $E_g$.  By Lemma \ref{lem:finite-volume-difference} we have \begin{align*}
		\left|B_{\nu}^M(F_g^M) - D_\nu^M(F_g^M) \right|^2 \leq  \left(B_\nu^M((F_g^M)^2) + D_\nu^M((F_g^M)^2) \right) \iint_{E_g} \frac{(\phi(\eta) - \phi(\eta + \delta_x))^2}{\phi(\eta) + \phi(\eta + \delta_x)} \,dQ_M(x,\eta)\,.
	\end{align*}

	By Cauchy-Schwarz we bound \begin{align}\label{eq:Delta-M-CS}
		|\Delta_{\nu}^M(F)|^2 \leq \EE_{g}\left[B_\nu^M((F_g^M)^2) + D_\nu^M(((F_g^M)^2)\right] \cdot \EE_g\left[ \iint_{E_g} \frac{(\phi(\eta) - \phi(\eta + \delta_x))^2}{\phi(\eta) + \phi(\eta + \delta_x)} \,dQ_M(x,\eta) \right]
	\end{align}
where we write $\EE_g$ for the expectation where $g \in L \backslash G$ is sampled according to $\PP_{L \backslash G}$.
	By Fubini's theorem we have \begin{align}
	\EE_{g}\left[ \iint_{E_g} \frac{(\phi(\eta) - \phi(\eta + \delta_x))^2}{\phi(\eta) + \phi(\eta + \delta_x)} \,dQ_M(x,\eta) \right] &= \iint_{M \times \Omega_M} \P_{L \backslash G}(Z_g(x) \neq \emptyset)  \frac{(\phi(\eta) - \phi(\eta + \delta_x))^2}{\phi(\eta) + \phi(\eta + \delta_x)} \,dQ_M(x,\eta) \nonumber \\
	&\leq \frac{C(F)}{\Vol_M(M)}I_M(\nu) \label{eq:ratio-I-bound}
	\end{align}
where the last inequality is by bounding $\P_{L \backslash G}(Z_g(x) \neq \emptyset) \leq \frac{C(F)}{\Vol_M(M)}$.    To bound the first term in \eqref{eq:Delta-M-CS}, note that \begin{equation*}
	F_g^M(x,\eta)^2 = \left(\sum_{z \in Z_g(x)} F(z,\theta_g \eta)\right)^2 \leq |Z_g(x)|^2 \|F\|_\infty^2 \leq C(F,r) |Z_g(x)| 
\end{equation*}
where we use \eqref{eq:Z-bound} for the last inequality.  This provides the bound \begin{align}\label{eq:death-M-bound}
	\EE_g  D_\nu^M((F_g^M)^2) \leq C(F,r) \int_{\Omega_M}\sum_{x \in \eta} \EE_g |Z_g(x)| \,d\nu(\eta) \leq   \frac{ C(F,r)'}{\Vol_M(M)}\int_{\Omega_M} |\eta| \cdot d\nu(\eta) \leq \lambda C(F,r)'
\end{align}
where the last bound is via the Ruelle bound in Lemma \ref{lem:mu_t-Poisson-factor}.  Similarly, we bound \begin{align}\label{eq:birth-M-bound}
	\EE_g B_\nu^M((F_g^M)^2) \leq \lambda C(F,r) \int_M \EE_g |Z_g(x)| \,d\Vol_M(x) \leq \lambda C(F,r)''
\end{align}
by the same argument.  Combining \eqref{eq:Delta-M-CS} with \eqref{eq:ratio-I-bound}, \eqref{eq:death-M-bound} and \eqref{eq:birth-M-bound} completes the proof.	
\end{proof}

To simplify notation, let $\Delta(M,\nu;F) = \Delta^M_\nu(F)$. We now compare $\Delta(M,\nu_t^M;F)$ to $\Delta_{\nu_t}(F)$; the engine for this is our bound on the propagation of information in Corollary \ref{cor:finite-propagation-backwards}.  

\begin{lemma}\label{lem:finite-propagation}
    Let $M_n = L_n \backslash \bX$ be a sequence of compact quotients so that $M_n$ Benjamini-Schramm converges to $\bX$.  For every $T < \infty$ and good test function $F$ we have $$\lim_{n\to\infty} \sup_{0 \leq t \leq T} \left|\Delta\left(M_n, \nu_t^{M_n};F\right) - \Delta_{\nu_t}(F)\right| = 0 \,.$$    
\end{lemma}
\begin{proof}
	Define $K^+ = \{x \in \bX: \dist(x,K) \leq R\}$.  Note that the functions \begin{equation*}
        \eta \mapsto \sum_{x \in \eta} F(x,\eta - \delta_x) \qquad \text{ and } \qquad  \eta \mapsto \int_{\bX} F(x,\eta) e^{-H(x\,|\,\eta)} \,d\Vol_{\bX}(x)
    \end{equation*}
    are uniformly bounded in $\eta$ since $F$ is a good test function and only depend on the configuration in $K^+$.  Let $C(F)$ denote a uniform bound on these two functions.  For a parameter $\rho > 0$ to be chosen later, define $K^{++} = \{x \in \bX: \dist(x,K^+) \leq \rho\}$.  As in Corollary \ref{cor:finite-propagation-backwards}, consider the backwards ancestor cluster of all points that lie in $K^+$ and whose birth time is at most $T$.  By Corollary \ref{cor:finite-propagation-backwards}, the probability this cluster leaves $K^{++}$ is at most \begin{equation}
    	\eps_\rho := \sum_{k \geq \rho / (2r)} \lambda^{k+1} \Vol_\bX(K^+)  (\Vol_{\bX}(B_{2r}(\mathcal{O})))^k \frac{T^{k+1}}{(k+1)!} 
    \end{equation}
which we note tends to zero as $\rho \to \infty$.  Further, this estimate holds uniformly on all quotients since the projected volumes of $K^+$ and $B_{2r}(\mathcal{O})$ are at most their unprojected volumes.  Define $s$ so that $K^{++} \subset B_s(\mathcal{O})$.    For a fixed $\eps > 0$, we may take $\rho$ large enough so that $\eps_\rho \leq \eps$.  For $g \in L_n \backslash G$,  if $\pi_g$ is injective on $B_s(\mathcal{O})$ then the marked point processes on $\bX$ and $M_n$ may be coupled so that they agree inside $K^{++}$ after identifying points via $\pi_g$.  If the backwards ancestor clusters that can affect $K^+ \times [0,T]$ remain inside $K^{++}$ then all acceptance decisions in $K^+$ agree for all $t \leq T$.  When both of these good events occur, the finite-volume and infinite-volume birth and death observables agree. 
Using a uniform bound of $C(F,\lambda)$ away from this event, we obtain the bound \begin{equation*}
	\sup_{0 \leq t \leq T} \left|\Delta\left(M_n, \nu_t^{M_n};F\right) - \Delta_{\nu_t}(F)\right|\leq 4 C(F,\lambda)\left( \P_{L_n \backslash G}(g : \pi_g \text{ is not injective on } B_s(\mathcal{O})) + 2\eps_\rho \right).
\end{equation*}
    Since the probability on the right-hand-side tends to $0$ as $n \to \infty$ by the assumption of Benjamini-Schramm convergence, we have the bound $$\limsup_{n \to \infty} \sup_{0 \leq t \leq T} \left|\Delta\left(M_n, \nu_t^{M_n};F\right) - \Delta_{\nu_t}(F)\right| \leq 8 C(F,\lambda) \eps_\rho\,.$$ 
	Taking $\rho \to \infty$ completes the proof.
\end{proof}

\subsection{Proof of Theorem \texorpdfstring{\ref{th:sparser-measure}}{32}}

We first show that the birth and death functionals are continuous as functions of the underlying measures.

\begin{lemma}\label{lem:continuity-BD}
    Let $\nu_j \in \M_r(\bX)$ be a sequence satisfying a uniform Ruelle bound and so that $\nu_j$ converges weakly to $\nu$ as $j\to\infty$.  Then for every good test function $F: \bX \times \Omega \to \R$  we have $$B_{\nu_j}(F) \to B_\nu(F) \text{ and } D_{\nu_j}(F) \to D_\nu(F)\,.$$
\end{lemma}
\begin{proof}
    The death functional is the expectation of the function $\eta \mapsto \sum_{x \in \eta} F(x,\eta - \delta_x)$.  This functional is bounded and continuous with compact support and so we immediately obtain $D_{\nu_j}(F) \to D_\nu(F)$.  For the birth functional, we may approximate the indicator $\one\{\dist(x,\eta) \geq 2r\}$ above and below by continuous functions that both lie between $\one\{\dist(x,\eta) \geq 2r \pm \eps\}$ for each fixed $\eps > 0$.  Using the Ruelle bound in Lemma \ref{lem:mu_t-Poisson-factor} along with noting that under Assumption \ref{assumption-on-X} we have $\lim_{\eps \to 0}\Vol(B_{2r+\eps}(\mathcal{O})) - \Vol(B_{2r-\eps}(\mathcal{O})) = 0$ shows the approximation error can be made arbitrarily small.  Applying weak convergence to these continuous approximations completes the proof. 
\end{proof}

We are now ready to prove Theorem \ref{th:sparser-measure}.

\begin{proof}[Proof of Theorem~\ref{th:sparser-measure}]
    Let $F_1,F_2,\ldots$ be a countable dense family of good test functions.  Let $n_j$ be large enough so that for all $i \leq j$ Lemma \ref{lem:finite-propagation} assures that for $M_j = M_{n_j}$ we have \begin{equation}\label{eq:proof-finite-propagation}
        \sup_{t \leq j} \left|\Delta\left(M_j, \nu_t^{M_j};F_i\right) - \Delta_{\nu_t}(F_i) \right| \leq j^{-1} \,.
    \end{equation}
    By Corollary \ref{cor:bound-on-information} there is  $t_j \in [j/2,j]$ so that \begin{equation*}
        I_{M_j}\left(\nu_{t_j}^{M_j}\right) \leq \frac{2 \lambda \Vol_{M_j}(M_j)}{j}\,.
    \end{equation*}
    Combining this estimate with Lemma \ref{lem:finite-volume-dissipation}  and \eqref{eq:proof-finite-propagation} shows \begin{equation} \label{eq:bound-GNZ-diff-finite-vol}
        \left|\Delta_{\nu_{t_j}}(F_i)\right| \leq j^{-1} +  \lambda C(F_i,r) j^{-1/2}
    \end{equation}
    for all $i \leq j$.  For each fixed $i$ we thus have  $$\lim_{j \to \infty} \Delta_{\nu_{t_j}}(F_i) = 0\,.$$
    By compactness, there is some subsequence of $(\nu_{t_j})_{j}$ along which we have a limiting measure $\nu_\infty.$  By Lemma \ref{lem:mu_t-Poisson-factor}, we have that $\nu_\infty$ is a weak Poisson factor; since each $\nu_{t_j}$ is $G$-invariant, so is $\nu_\infty$. By Lemma \ref{lem:continuity-BD} we have $\Delta_{\nu_\infty}(F_i) = 0$ for all $i$, which shows that $\nu_\infty \in \Gibbs(\bX;r,\lambda)$ by the GNZ equations \eqref{eq:GNZ-equations}.  
\end{proof}

\section{Lattice actions are not weak Poisson factors}
\label{secLattice}

Recall by Theorem \ref{T:special} that at the tight radii $\{r_n\}_{n \geq 7} \cup \{\infty\}$ from Section \ref{secLatticeMeasures} we have unique optimally dense invariant measures $\mu_n \in \M_{r_n}(\H^2)$ coming from lattice subgroups $L_n \leq \Isom(\H^2)$.  The main goal of this section is to prove that the measures $\mu_n$ are not weak Poisson factors.  This is immediate from the following.

\begin{theorem}\label{thm:periodic-not-Poisson-factor}
    Suppose that $\mu \in \M_\rho(\H^d)$ (for some $\rho$) or $\mu \in \Prob_{\isom} (\Closed(\H^d))$ and there is a packing $P$ such that
 \begin{itemize}
     \item $\mu(\Isom(\H^d)\cdot P)=1$ ($\mu$ is concentrated on the $\Isom(\H^d)$-orbit of $P$);
     \item $\Stab(P)=\{g\in \Isom(\H^d):~gP=P\}$ contains a lattice subgroup. 
 \end{itemize}   
    Then $\mu$ is not a weak Poisson factor.  
\end{theorem}

We will in fact prove the following generalization.

\begin{theorem}\label{T:weak}
Let $G$ be an lcsc group with a discrete non-abelian free group $F<G$ and a lattice subgroup $L<G$. For $f\in F$, let $\phi_f:G/L \to G/L$ be given by left-translation: $\phi_f(gL)=fgL$. Assume that every $\phi_f$ has finite entropy and there exists an $f_0\in F$ such that $\phi_{f_0}$ is ergodic with respect to Haar measure. 

Suppose $G \cc X$ is a jointly continuous action on a Polishable space (that is, $X$ admits a complete separable metric inducing its topology), $\nu \in \Prob(X)$ is an invariant probability measure and $\nu(Gx)=1$ for some $x\in X$ with stabilizer $L$ (this means $L=\{g\in G:~ gx=x\}$). Then $\nu$ is not a weak Poisson factor. 
\end{theorem}

To deduce Theorem~\ref{thm:periodic-not-Poisson-factor} from Theorem~\ref{T:weak}, we will need to identify the free group $F$, show it contains an element that acts ergodically, and show that each action has finite entropy.  We use the next result for this.

\begin{prop}\label{C:weak}
  If $G$ is a semi-simple Lie group  and $L \le G$ is an irreducible  lattice, then there is a free group $F \le G$ satisfying the hypotheses of Theorem \ref{T:weak}. Therefore, the action $G \cc G/L$ is not measurably conjugate to a weak Poisson factor (with respect to Haar measure on $G/L$). 
\end{prop}

\begin{proof}
    Because $G$ is non-amenable, the lattice $L$ is also non-amenable. It then follows from the Tits alternative for Lie groups \cite{MR286898} that $L$ contains a non-abelian finite rank free group $F$.     Because $L$ is discrete, $F$ is also discrete. The Howe-Moore Theorem\footnote{For an exposition of the Howe-Moore Theorem, see \cite[Theorem 2.2.20]{zimmer} or \cite[Theorem 1.1]{MR1781937}.} implies that every nontrivial element $f\in F$ acts ergodically on  $G/L$.   
    
Rufus Bowen proved that the topological entropy $h_d(T)$ of a diffeomorphism $T$ of a locally compact $m$-dimensional Riemannian manifold $(M,d)$ is bounded by $\max(0, m \log \sup_{x \in M} \|dT|T_xM\|)$ \cite[Proposition 12]{MR274707}. We apply this to the case $M=G/L$ and $T=\phi_f$ (so $T(gL)=fgL$ for some fixed $f \in F \setminus \{e\}$) to obtain that the topological entropy of $\phi$ is finite.  Handel and Kitchens proved\footnote{Note that we are not assuming $G/L$ is compact and so the well-known Variational Principle does not apply.}  that the topological entropy is an upper bound for the measure-theoretic entropy \cite[Proposition 1.4]{MR1348316} which completes the proof. 
\end{proof}

\begin{proof}[Proof of Theorem \ref{thm:periodic-not-Poisson-factor} from Theorem \ref{T:weak} and Proposition \ref{C:weak}]
    By \cite[Theorem 5.1.5]{beer1993topologies}, $\Closed(\H^d)$ is a Polishable  space with respect to the Fell topology.  Also $\cP_\rho(\H^d)$ is a subset of the space of Radon measures on $\H^d$, which may be identified with a subset of the Banach dual $C_0(\H^d)^*$ via the Riesz Representation Theorem.  Under the weak* topology, $\cP_\rho(\H^d)$ is a closed subset of $C_0(\H^d)^*$ and thus $\cP_\rho(\H^d)$ is a Polishable space since there exists a complete separable metric inducing its topology.  

  Note that for $d\geq 2$ we have that $G = \Isom(\H^d)$ is semi-simple. Since it has rank 1, all lattices in $G$ are irreducible.  
    Applying Proposition \ref{C:weak} and Theorem \ref{T:weak} shows that $\mu$ is not a weak Poisson factor.
\end{proof}

For the rest of this section, we fix a group $G$ as in the statement of Theorem \ref{T:weak}. The first part of the proof reduces the problem to showing that a smooth action of a free group $F$ cannot be a weak Bernoulli factor.

\subsection{Reduction to free groups} \label{sec:reduction-free}

We assume the hypotheses of Theorem \ref{T:weak}.  By taking a subgroup, we can  assume that $F$ is finitely generated. We will use the Poisson factor to identify a \emph{Bernoulli factor}, a more classical and well-studied notion in the dynamical systems literature.  To begin with, we recall the definition of Bernoulli shifts.

\begin{definition}
    Let $\G$ be a countable group and $(\mathtt{K},\kappa)$ be a standard Borel probability space. Let $(\mathtt{K},\kappa)^\G = (\mathtt{K}^\G,\kappa^\G)$ denote the product probability space. The action $\G \cc (\mathtt{K},\kappa)^\G$ defined by $(gx)(f)=x(g^{-1}f)$ is the \textbf{Bernoulli shift over $\G$ with base space $(\mathtt{K},\kappa)$.} 
\end{definition}

We now recall the definition of a Bernoulli factor.

\begin{definition}
Suppose $\G$ acts on a Polishable space $X$ and $\nu$ is a $\G$-invariant Borel probability measure on $X$. Then we say $\nu$ is a \textbf{Bernoulli factor} if there are a Bernoulli shift $\G \cc (\mathtt{K},\kappa)^\G$ and a factor map $\Phi:\mathtt{K}^\G \to X$ such that $\Phi_*\kappa^\G = \nu$. We say $\nu$ is a \textbf{weak Bernoulli factor in $X$} if it is a weak$^*$ limit of Bernoulli factors. This means there exists a sequence $\nu_1,\nu_2,\ldots \subset \Prob(X)$ of Bernoulli factors which limit on $\nu$ in the weak$^*$ topology.

\end{definition}

We recall the notion of measurably conjugate actions, which is a notion of isomorphism of dynamical systems. 

\begin{definition}
    Two probability measure preserving (pmp) actions $\G \cc (X,\mu)$ and $\G \cc (Y,\nu)$ are \textbf{measurably conjugate} if there exist conull sets $X' \subset X, Y' \subset Y$ and an invertible measurable map $\Phi:X' \to Y'$ such that $\Phi(gx)=g\Phi(x)$ for all $g\in \G$, $x\in X'$ such that $\Phi_*\mu=\nu$.
\end{definition}

We first identify the $F$ action on $(\cP(G),\Pois(\lambda_G))$ as a Bernoulli shift.

\begin{lemma}\label{lem:mc-to-Bernoulli-shift}
Let $\lambda_G$ be a left-Haar measure on $G$ and let $F < G$ be a discrete non-abelian free subgroup. Then the action of $F$ on $(\cP(G), \Pois(\lambda_G))$ is measurably conjugate to a Bernoulli shift.   
\end{lemma}

\begin{proof}
  Because $F$ is discrete, it admits a measurable fundamental domain $\Delta \subset G$, meaning that $G$ is the disjoint union of the translates $f\Delta$ over $f\in F$ \cite[Lemma 4.1.1]{MR3307755}. 

    Let $\mathtt{A}=\cP(\Delta)$ be the space of point measures on $\Delta$, $\lambda_\Delta$ be the restriction of $\lambda_G$ to $\Delta$ and $\kappa=\Pois(\lambda_\Delta)$ be the law of the Poisson point process on $\Delta$. Define $\Phi:\cP(G)\to \mathtt{A}^F$ by requiring that $\Phi$ is the unique $F$-equivariant map which satisfies $\Phi(\Pi)(1_F)$ is the restriction of $\Pi$ to $\Delta$. Then $\Phi$ is a measure-conjugacy from $F \cc (\cP(G),\Pois(\lambda))$ to $F \cc (\mathtt{A},\kappa)^F$.
        \end{proof}

With Lemma \ref{lem:mc-to-Bernoulli-shift} in hand, we now show that our weak Poisson factor can be realized as a weak Bernoulli factor.

\begin{corollary}\label{C:reductiontofree}
Let $X$ be a Polishable space and $G \cc X$ a jointly continuous action.    If $\mu \in \Prob(X)$ is a weak Poisson factor for the $G$-action then $\mu$ is also a weak Bernoulli factor for the induced $F$-action. 
\end{corollary}

\begin{proof}
Suppose $\mu_1,\mu_2,\ldots \in \Prob(X)$ are Poisson factors 
which weak$^*$ converge to $\mu$. So for each $i$ there exist a compact subgroup $K_i \le G$, a $G$-invariant measure $\lambda_i$ on $G/K_i$ and a $G$-equivariant factor map $\Phi_i:\cP(G/K_i) \to X$ such that $\Phi_{i*}\Pois(\lambda_i) = \mu_i$. By Lemma \ref{lem:mc-to-Bernoulli-shift}, $\Phi_i$ demonstrates that $\mu_i$ is a Bernoulli factor with respect to the $F$-action. Therefore $\mu$ is a weak Bernoulli factor.
\end{proof}

\subsection{Annealed entropy and free groups} \label{sec:annealed-entropy}

In light of Corollary \ref{C:reductiontofree}, to finish the proof of Theorem \ref{T:weak} it suffices to show that no smooth action of $F$ is a weak Bernoulli factor. By a smooth action, we mean an action by diffeomorphisms on a differentiable manifold; in our case, the differentiable manifold will be $M = G/L$.

We will use the entropy theory of free group actions as developed by the first-named author in \cite{MR2630067, MR2653969}. We will define the annealed entropy (also known as the $f$-invariant) of a free group action, prove that it is negative on smooth actions 
but non-negative on weak Bernoulli factors. This will imply Theorem \ref{T:weak}.

To begin, let $\mathtt{A}$ denote a finite set which we call the alphabet. Then $\mathtt{A}^F$ is the space of all functions $x:F \to \mathtt{A}$ with the topology of uniform convergence on finite sets (which is the same as the product topology). In this topology $\mathtt{A}^F$ is compact. Let $F$ act on $\mathtt{A}^F$ by
$$(fx)(g)=x(f^{-1}g)$$
for $f,g \in F$ and $x\in \mathtt{A}^F$.

We will define the annealed entropy of $F$-invariant measures on $\mathtt{A}^F$. To do this, let $\Hom(F,\Sym(n))$ be the set of homomorphisms from $F$ to the symmetric group $\Sym(n)$ on $[n]$. Given $\s \in \Hom(F,\Sym(n))$, $x:[n] \to \mathtt{A}$ and $v\in [n]$ define the \textbf{pull-back name of $x$ based at $v$} by
$$\Pi^\s_v(x) \in \mathtt{A}^F, \quad \Pi^\s_v(x)(g) = x(\s(g^{-1})v).$$
This is defined so that the map from $[n]$ to $\mathtt{A}^F$ defined by $v \mapsto \Pi^\s_v(x)$ is $F$-equivariant in the sense that $\Pi^\s_{\s(g)v}(x) = g\cdot \Pi^\s_v(x)$ for $g\in F$.

The \textbf{empirical measure of $x$ with respect to $\s$} is the law of $\Pi^\s_v(x)$ where $v$ is chosen uniformly at random in $[n]$. In symbols,
$$P^\s_x = n^{-1} \sum_{v \in [n]} \delta_{\Pi^\s_v(x)} \in \Prob(\mathtt{A}^F)$$
where $\delta_y$ denotes the Dirac measure at $y$.

Intuitively, the annealed entropy of an $F$-invariant measure $\mu \in \Prob(\mathtt{A}^F)$ is the  exponential rate of growth of the expected number of microstates $x:[n]\to \mathtt{A}$ such that the empirical measure of $x$ with respect to $\s$ is a good approximation to $\mu$, where $\s$ is chosen uniformly at random. To make this precise, let $\cU \subset \Prob(\mathtt{A}^F)$ be an open set. Let $\Omega(\sigma,\cU)$ be the set of all $x\in \mathtt{A}^n$ whose empirical measure $P^\s_x \in \cU$. Let $\EE_n$ be the expected value with respect to the uniform probability measure on  $\Hom(F,\sym(n))$. Then the \textbf{annealed entropy} of an $F$-invariant measure $\mu \in \Prob(\mathtt{A}^F)$ is defined by
$$h^{\ann}(\mu) = \inf_{\cU} \limsup_{n\to\infty} n^{-1}\log \EE_n[\#\Omega(\sigma,\cU)]$$
where the first infimum is over all open neighborhoods of $\mu$ in $\Prob(\mathtt{A}^F)$.

In previous papers, this entropy is called the \textbf{$f$-invariant} and denoted $f(\mu)$ instead. This invariant was introduced in \cite{MR2630067} via an alternative (but equivalent) formula. In that paper, $h^{\ann}$ is proven to be invariant under measure-conjugacy. The paper \cite{MR2653969} proves that the original definition is equivalent to the definition given above.  The annealed entropy is also occasionally referred to as annealed sofic entropy.

To help understand this definition, note that if we pick an element $\Hom(F,\Sym(n))$ uniformly at random this is equivalent to choosing random permutations $s_1,\ldots,s_r \in \Sym(n)$ independently and uniformly at random for the $r$ generators in $F$.  If one considers the graph with vertices $[n]$ and edges of the form $(v,s_i\cdot v)$, then as $n$ tends to infinity this sequence of random graphs Benjamini-Schramm converges to the Cayley graph of $F$ with the usual generators and their inverses.  A sequence of graphs that Benjamini-Schramm converges to $F$ is called a sofic approximation. The quantity above is the annealed sofic entropy since we are averaging over the sofic approximation; if one were to instead fix a sofic approximation, and look at the limit along such a sequence, the obtained quantity is called the quenched sofic entropy, or occasionally simply the sofic entropy.

At the moment, we have only defined annealed entropy for probability measures on $\mathtt{A}^F$.  We now introduce the notion of a finitely generated action, which will allow us to extend the definition of annealed entropy to a wider class of actions.  

\begin{definition}
Let $\G \cc (X,\mu)$ be a pmp action of a countable group $\G$. The action is said to be \textbf{finitely generated} if there is a measurable map $\phi:X \to \mathtt{A}$ (where $\mathtt{A}$ is a finite set) such that the smallest complete $\G$-invariant sigma-algebra containing $\phi^{-1}(a)$ for all $a\in \mathtt{A}$ is the sigma-algebra of all measurable sets.

Note that if such a map $\phi$ exists then we may define $\Phi:X \to \mathtt{A}^\G$ by $\Phi(x)(g) = \phi(g^{-1}x)$. This map is $\G$-equivariant and, because $\phi$ is generating, it is 1-1 almost everywhere and therefore defines a measure-conjugacy between $\G \cc (X,\mu)$ and its image $\G \cc (\mathtt{A}^\G, \Phi_*\mu)$. Because annealed entropy is invariant under measure-conjugacy, this allows us to define the annealed entropy of any finitely generated pmp action of the free group $F$.
\end{definition}

\begin{proposition}\label{P:weakBernoullifactor1}
Suppose $X$ is a Polishable space, $F \cc X$ is a jointly continuous action, $\mu$ is an $F$-invariant measure on $X$ and $F \cc (X,\mu)$ is finitely generated. If $\mu$ is a weak Bernoulli factor in $X$ then there exist a finite alphabet $\mathtt{A}$ and an $F$-invariant measure $\nu \in \Prob(\mathtt{A}^F)$ such that 
\begin{enumerate}
    \item $F \cc (X,\mu)$ is measurably conjugate to $F \cc (\mathtt{A}^F, \nu)$;
    \item $\nu$ is a weak Bernoulli factor in $\mathtt{A}^F$. 
\end{enumerate}
\end{proposition}

\begin{proof}
    For $i=1,2,\ldots$, let $\mu_i$ be $F$-invariant Borel probability measures on $X$ which converge to $\mu$ such that each $\mu_i$ is a Bernoulli factor. Because $\mu$ is finitely generated, there is a measurable function $\phi:X \to \mathtt{A}$ where $\mathtt{A}$ is a finite set such that if $\Phi:X \to \mathtt{A}^F$ is defined by $\Phi(x)(g)=\phi(g^{-1}x)$, then $\Phi$ is a measure-conjugacy from $F \cc (X,\mu)$ to $F \cc (\mathtt{A}^F, \Phi_*\mu)$. Set $\nu = \Phi_*\mu$.

So it suffices to prove that for every open neighborhood $\cU$ of $\nu$ in $\Prob(\mathtt{A}^F)$, there exists a Bernoulli factor $\nu' \in \cU$. Because the topology of $\mathtt{A}^F$ is generated by cylinder sets, it suffices to show that for every finite subset $H \subset F$ and $\epsilon>0$ there exists a Bernoulli factor $\nu'$ such that for every function $\tau:H \to \mathtt{A}$, if 
$$C_\tau = \{x\in \mathtt{A}^F:~ x(h)=\tau(h)~\forall h \in H\}$$
then 
$$|\nu( C_\tau) - \nu'(C_\tau)| < \epsilon.$$

Recall that if $Y \subset X$ then its topological boundary is defined by
$\partial Y = \bar{Y} \cap \overline{X\setminus Y}.$
    Moreover, $Y$ is a \textbf{continuity set for $\mu$} if $\mu(\partial Y)=0$.  It is well-known that continuity sets are dense in the measure-algebra\footnote{To see this, note first that for a fixed $x\in X$ and Lebesgue almost all $r$ we have that balls of radius $r$ centered at $x$ are continuity sets since $r \mapsto \mu(B(x,r))$ is monotone.  We may then construct a countable base with null boundaries using these balls and define the algebra of Borel sets that are approximable by these balls. A monotone class argument then approximates all Borel sets in the measure-algebra. }. 
    This means: for every measurable set $Y \subset X$ and $\epsilon>0$ there exists a $\mu$-continuity set $Y'$ such that $\mu(Y \vartriangle Y')<\epsilon$ where $\vartriangle$ denotes symmetric difference. 

    Now fix $H$ and $\epsilon$. Because continuity sets are dense in the measure-algebra there exists a measurable function  $\psi:X \to \mathtt{A}$ such that
    \begin{enumerate}
        \item for every $a\in \mathtt{A}$, $\psi^{-1}(a) \subset X$ is a $\mu$-continuity set;
        \item $|\mu(C(\phi,\tau)) - \mu(C(\psi,\tau))|<\epsilon$ for every function $\tau: H \to \mathtt{A}$, where
        $$C(\psi,\tau) =\{x\in X:~ \psi(hx)=\tau(h) ~\forall h \in H\}$$
        $$C(\phi,\tau) =\{x\in X:~ \phi(hx)=\tau(h) ~\forall h \in H\}.$$
    \end{enumerate}
Define $\Psi: X \to \mathtt{A}^F$ by $\Psi(x)(g)=\psi(g^{-1}x)$. Then $\Psi_*\mu_i$ is a Bernoulli factor  and $\Psi_*\mu_i$ converges in the weak* topology to $\Psi_*\mu$ which is in $\cU$ by construction. It follows that if $i$ is sufficiently large then $\Psi_*\mu_i$ is a weak Bernoulli factor in $\cU$, as required.
\end{proof}

We now are in place to begin to separate the probability measure $\nu$ from Theorem \ref{T:weak} from a weak Poisson factor.  We first show that weak Bernoulli factors have non-negative annealed entropy.

\begin{proposition}\label{P:nonnegative}
Suppose $X$ is a Polishable space, $F \cc X$ is a jointly continuous action, $\mu$ is an $F$-invariant measure on $X$ and $F \cc (X,\mu)$ is finitely generated. If $\mu$ is a weak Bernoulli factor, then $h^{\ann}(\mu) \ge 0$.
\end{proposition}

\begin{proof}
By Proposition \ref{P:weakBernoullifactor1}, we may assume without loss of generality that $X=\mathtt{A}^F$ for some finite alphabet $\mathtt{A}$.  Thus there exist Bernoulli factors $\mu_i \in \Prob(\mathtt{A}^F)$ which converge to $\mu$ as $i\to\infty$. By \cite[Theorem 3.2]{kerr-cpe}, each $\mu_i$ has positive annealed entropy.\footnote{It is much easier to prove that $\mu_i$ has non-negative annealed entropy which is all that we need. However, a short proof of that does not appear to be in the literature.}

Annealed entropy is upper semi-continuous by \cite{MR2630067}. This implies the proposition.
\end{proof}

We now show that finite entropy actions have annealed entropy equal to $- \infty$.

\begin{proposition}\label{P:smoothactions}
Let $F \cc (M,\mu)$ be a pmp action where $M = G/L$. For $f \in F$, let $\alpha_f:M \to M$ be the map $\alpha_f(x)=fx$. Suppose $\alpha_f$ has finite entropy for every $f$.
    \begin{enumerate}
        \item If there exists $f \in F$ such that $\alpha_f$ is ergodic with respect to $\mu$ then the action is finitely generated.
        \item   If $\mu$ is non-atomic and the action is finitely generated then $h^{\ann}(\mu)=-\infty$. 
    \end{enumerate}

\end{proposition}

\begin{proof}

It follows from Krieger's Theorem \cite{MR259068} that finite entropy and ergodicity implies finite generation. The second statement uses the following alternative formula for annealed entropy. Let $\phi:X \to \mathtt{A}$ be an observable with finite range. For $W \subset F$, define $\phi^W:X \to \mathtt{A}^W$ by
$$\phi^W(x)(w)=\phi(wx).$$

Given an element $f\in F$, let $h(\alpha_f,\phi,\mu)$ be the classical Kolmogorov-Sinai (KS) entropy with respect to the observable $\phi$:
$$h(\alpha_f,\phi,\mu)=\lim_{n\to\infty} \frac{1}{n} H_\mu\left(\phi^{\{e,f,\ldots, f^{n-1}\}}\right).$$
Also let $h(\alpha_f,\mu) = \sup_\phi h(\alpha_f,\phi,\mu) $ be the KS-entropy. 

We may assume $F=\langle s_1,\ldots, s_r\rangle$ is a free group of rank $r$ for some $2\le r < \infty$. Then \cite[Theorem 9.1]{MR2736889}
 shows that
 \begin{align}\label{E:f*}
     h^{\ann}(\mu)=\lim_{n\to\infty} (1-r) H_\mu(\phi^{B_n})+\sum_{i=1}^r h(\alpha_{s_i},\phi^{B_n},\mu)
 \end{align}
where $B_n$ is the ball of radius $n$ in $F$ with respect to the word metric. To be precise, \cite{MR2736889} shows that the right-hand side is the $f$-invariant, which is proven to be equal to the annealed entropy in \cite{MR2653969}.

By hypothesis, 
$$h(\alpha_{s_i},\phi^{B_n},\mu) \le h(\alpha_{s_i},\mu)<\infty$$
for each $i$. Because $\phi$ is generating, $\phi^{B_n}$ increases to the full sigma-algebra. So $H_\mu(\phi^{B_n})$ limits on 
the Shannon entropy of $(X,\mu)$. Because $\mu$ is non-atomic, this Shannon entropy is $+\infty$. So \eqref{E:f*} shows $h^{\ann}(X,\mu)=-\infty$ as claimed.
\end{proof}

\begin{proof}[Proof of Theorem \ref{T:weak}]
By assumption $\nu$ is supported on the $G$-orbit of an element $x$ of $X$ whose $G$-stabilizer is a lattice $L$. Therefore, the action $G \cc (X,\nu)$ is measurably conjugate to $G \cc (G/L, \lambda_{G/L})$ where $\lambda_{G/L}$ is the unique $G$-invariant measure on $G/L$.

By assumption, $G$ contains a discrete non-abelian free subgroup $F$ which satisfies the hypotheses of Proposition \ref{P:smoothactions} with $M=G/L$. Therefore $h^{\ann}(\nu)=-\infty$.  By Proposition \ref{P:nonnegative}, $\nu$ cannot be a weak Bernoulli factor with respect to the action of $F$. By Corollary \ref{C:reductiontofree}, $\nu$ cannot be a weak Poisson factor.

\end{proof}

\section{Proofs of the main theorems}
\label{secPoissonGap}

 Recall that an invariant measure $\mu \in \M_r(\H^d)$ is an {\bf optimally dense measure on radius $r$ packings} if $\density_r(\mu)=D_{\opt}(\H^d,r)$. Similarly, if $\mu$ is an invariant measure on the space of horoballs with $\density(\mu)=D_{\opt}(\H^d,\infty)$, then we say $\mu$ is {\bf optimally dense at radius $\infty$}.

\begin{lemma}\label{L:PoissonGap1}
    Fix $\rho \in (0,\infty]$. Suppose there exists a unique optimally dense measure on radius $\rho$ packings and it is periodic. Then
        $D_{\Pois}(\H^d,\rho)<\Dopt(\H^d,\rho)$.
\end{lemma}
\begin{proof}
 Seeking a contradiction, assume that $D_{\Pois}(\H^d,\rho)  = \Dopt(\H^d,\rho)$. Then there is a sequence of Poisson factors $\mu_n$ so that $\density_\rho(\mu_n) \to \Dopt(\H^d,\rho)$.  By passing to a subsequence, this implies that there is a weak Poisson factor $\mu$ with $\density_\rho(\mu) = \Dopt(\H^d,\rho)$.  By the assumption of uniqueness of the optimally dense packing, we have that $\mu$ must be a periodic weak Poisson factor.  This contradicts Theorem \ref{thm:periodic-not-Poisson-factor}. 
\end{proof}

We now show that there is a gap between optimal densities achieved by Poisson factors and optimal densities in dimension $2$.

\begin{proof}[Proof of Theorem \ref{T:poissongap}]
Given a dimension $d\ge 2$, let $R_d \subset (0,\infty]$ be the set of all radii $\rho \in (0,\infty]$ such that $D_{\Pois}(\H^d,\rho)<\Dopt(\H^d,\rho)$. By Lemma \ref{L:PoissonGap1} and Theorems \ref{T:special}, \ref{T:horoball1}, $R_2 \supset \{r_n\}_{n\ge 7} \cup \{\infty\}$. Because $D_{\opt}(\H^2,\cdot)$ is continuous (Theorem \ref{T:continuity2}) and $D_{\Pois}(\H^2,\cdot)$ is left-continuous and upper semi-continuous (Lemma \ref{L:PoissonGap2}), it follows that $R_2$ is open in $(0,\infty]$. This implies the theorem.
\end{proof}

Due to our construction of a Gibbs measure with density bounded by $D_{\Pois}$, we now deduce that non-uniqueness occurs in $\H^2$:

\begin{proof}[Proof of Theorem \ref{T:main}]
    As in the proof of Theorem \ref{T:poissongap}, let $R_2 \subset (0,\infty]$ be the set of all radii $\rho$ such that $D_{\Pois}(\H^2,\rho)<\Dopt(\H^2,\rho)$. It follows from Theorems \ref{T:near-optimal} and \ref{th:sparser-measure} that if 
$r \in R_2$ then $\lambda_u(\H^2,r)<\infty$. This implies the theorem.
\end{proof}

The existence of the gap between Poisson factors and optimal packings will be the main engine to prove that there exist completely saturated packings with suboptimal density.  Our proof will in fact show that any weak Poisson factor $\mu$ with $\density_r(\mu) = D_{\Pois}(\H^2,r)$ is supported on completely saturated packings.  The main idea of the proof is to show that if it is not completely saturated, one may alter $\mu$ to bump up its density, thereby contradicting optimality among weak Poisson factors.

\begin{proof}[Proof of Theorem \ref{T:saturated}]
This proof is essentially the same as the proof from \cite[Theorem 3.1]{bowen2003existence} that if $\mu \in \M_r(\H^d)$ is optimally dense then $\mu$-a.e. packing is completely saturated. So we will provide a sketch only. For this sketch, we will abuse terminology slightly by referring to elements of $\cP_r(\H^d)$ as packings rather than as $2r$-separated point measures.

Let $\cP^{\mathrm{cs}}_r(\H^d) \subset \cP_r(\H^d)$ be the subset of completely saturated packings. Let $\mu \in \M_r(\H^d)$ be a probability measure which is a weak Poisson factor and $\density_r(\mu)=D_{\Pois}(\H^d,r)$. To obtain a contradiction, suppose that $\mu(\cP^{\mathrm{cs}}_r(\H^d))<1$.

For $R>0$, let $X_R \subset \cP_r(\H^d)$ be the set of packings $P$ such that there exist  packings $P',P_0,P_1 \in \cP_r(\bX)$ satisfying
\begin{enumerate}
    \item $P$ is the disjoint union of $P'$ and $P_0$;
    \item $P_0$ and $P_1$ each contain only finitely many disks and $P_0$ contains fewer disks than $P_1$; 
    \item $E(P)=P' \cup P_1 \in \cP_r(\bX)$ is a packing;
    \item $P_0\cup P_1$ is contained in the ball of radius $R$ centered at the origin. 
\end{enumerate}
Because the complement of $\cP^{\mathrm{cs}}_r(\H^d)$ is equal to $\cup_{R=1}^\infty X_R$, there exists a radius $R>0$ such that $\mu(X_R)>0$. 

By definition, there is a function $E:X_R \to \cP_r(\H^d)$ satisfying the conditions mentioned above. This function can be chosen to be Borel (see, e.g., \cite[Lemma 4.1]{bowen2003existence}). Then we can choose a Borel $\Isom(\H^d)$-equivariant map $\Phi:\cP_r(\H^d) \to \cP_r(\H^d)$ so that the restriction of $\Phi(P)$ to the ball of radius $R$ centered at the origin is either equal to $P$ or $E(P)$ restricted to the same ball.

By isometry invariance and linearity of expectation, the density of an element of $\M_r(\H^d)$ may be realized as the expected proportion of the ball of radius $R + r$ that is covered. It follows that $\Phi_*\mu$ has greater density than $\mu$. However, $\Phi_*\mu$ is also a weak Poisson factor because the property of being a weak Poisson factor is closed under factors. This contradicts the optimality of $\mu$.
\end{proof}

\section{Open questions}\label{S:open}

We end with some open problems, many of which are geared towards understanding how one may prove phase transitions beyond $\H^d$.

\subsection{Densities and packings}
A key fact to our approach for large radius is that $\DPois(\H^2,\infty) < \Dopt(\H^2,\infty)$ by using uniqueness and periodicity of the optimizer of $\Dopt(\H^2,\infty)$.  It may be the case that this inequality follows from showing that $\DPois$ tends to zero. \begin{question}\label{q:dpois-infinity}
Let $d \geq 2$.  Do we have $\lim_{r \to \infty}\DPois(\H^d,r) = 0$?
\end{question}

Recall that \cite{bowen2003periodic} proves continuity of $\Dopt(\H^2,r)$.  It is natural to ask if the same holds for $\DPois$.  

\begin{question}
    Let $d \geq 2$.  Is it the case that $r \mapsto \DPois(\H^d,r)$ is continuous?  
\end{question}

Similarly, we suspect that the case of $d = 3$ is a natural next step for continuity of $\Dopt$.

\begin{question}
    Is $r \mapsto \Dopt(\H^3,r)$ continuous? Further, do we have $\Dopt(\H^3,r) = \Dper(\H^3,r)$?
\end{question}

It may be possible to extend the branched covering techniques of \cite{bowen2003existence} in order to prove continuity of $\Dopt(\H^3,\cdot)$; if so, then it may also be possible to prove $\Dopt(\H^3,r) = \Dper(\H^3,r)$ in a similar manner to the $2$-dimensional case since lattices in $\H^3$ have property MD \cite{bowen-tucker-2013}.

While continuity of   $\Dopt$ is interesting in all dimensions, the problem seems quite challenging in arbitrary dimensions.  It would already be interesting to understand if there is continuity at infinity. \begin{question}
    For $d \geq 3$ do we have $\lim_{r \to \infty} \Dopt(\H^d, r) = \Dopt(\H^d,\infty)$?
\end{question}

Theorem \ref{T:continuity2} proves the $d = 2$ case.  The case of $d = 3$ may be more tractable since $\Dopt(\H^3,\infty)$ is known \cite{boroczky1978packing} to match the simplex bound (although not uniquely \cite{kozma2012optimally}).

In order to extend our results to, say, $\H^3$ for all large $r$, one option is to understand the structure of the optimal horoball packings. \begin{question}
    Is it the case that all optimal horoball packings in $\H^3$ are periodic?  
\end{question}

If so, then it follows from Theorem \ref{thm:periodic-not-Poisson-factor} that $D_{\Pois}(\H^3,\infty)<D_{\opt}(\H^3,\infty)$.  It then appears that much of our work can be used to prove a phase transition in $\H^3$ for large $r$.

\subsection{Gibbs measures and probability}

In the case of independent sets on graphs, a rich literature studies the largest independent set achieved by a Bernoulli factor on the infinite regular tree \cite{gamarnik2017limits,rahman2017local}, {with applications to proving limits on certain classes of algorithms}.  One of the key tools used in this area is  understanding the geometry of the collection of independent sets on random regular graphs, since random regular graphs Benjamini-Schramm converge to the infinite regular tree.  In analogy to this, one may seek to understand the structure of radius $r$ packings on random hyperbolic surfaces.

\begin{question}
    Let $M_n$ be either given by: a hyperbolic surface of genus $n$, sampled at random according to Weil-Peterson measure on the moduli space of hyperbolic surfaces; or as a random degree $n$ cover of a fixed hyperbolic surface of genus at least $2$.  Each of these two cases has been shown to Benjamini-Schramm converge to $\H^2$ (\cite{monk2022benjamini,magee2023asymptotic}).  
    
    Consider the maximum density of a radius $r$ disk packing on $M_n$.  Does this density converge in probability as {$n \to\infty$}?   If so, then does this limit converge to $0$ as $r \to \infty$? Note that an affirmative answer to the latter question for either ensemble would confirm Question \ref{q:dpois-infinity} for $d = 2$.
\end{question}

Turning to the space of Gibbs measures, note that a curious corollary of the proof of Theorem \ref{T:main} is that we prove that density is not a function of activity for the hard-sphere model in the hyperbolic plane. This motivates the following problem: 
\begin{question}
    Given radius $r>0$ and density $\delta \in (0,D_{\opt}(\H^2,r))$, is there only one Gibbs measure on $\cP_r(\H^2)$ with density $\delta$? 
\end{question}

Additionally, our proof considers the isometry-invariant Gibbs measures on $\cP_r(\H^2)$.  It is known that there are non-invariant Gibbs measures for the hardcore model on the $\Delta$-regular tree~\cite{spitzer1975markov,kelly1985stochastic} and $\mathbb Z^d$, $d \ge 2$~\cite{dobrushin1968problem}.  In analogy with the discrete case, we ask: 
\begin{question}
    Does there exist a non-isometry-invariant Gibbs measure on $\cP_r(\H^2)$ for sufficiently large $\lam$? In particular, do the Gibbs measures constructed in this paper have trivial tail-sigma algebras? If one of them does not have a trivial tail-sigma algebra then it is not extremal in the space of Gibbs measures and this implies the existence of non-isometry-invariant Gibbs measures.
    
\end{question}

\begin{question}
    In Section \ref{sec:lower-density} we use  spatial birth-death Glauber dynamics starting from the empty packing to obtain a $(r,\lam)$-Gibbs measure which is also a weak Poisson factor. There we have to take a subsequential limit since it is unclear whether the limit exists.  Does the limit actually exist? If so, it would provide a natural class of examples. 
\end{question}

\section*{Acknowledgments}

This project began at a Graduate Mini-School at UT Austin which was funded by NSF DMS--1937215. L.B. is supported in part by NSF grants DMS--2154680 and DMS--2453399.
M.M.\ is supported in part by NSF grants DMS-2336788 and DMS-2246624.  W.P.\ supported in part by NSF grant DMS-2348743.

\bibliographystyle{abbrv}
\bibliography{spheres}

\end{document}